\documentclass[10pt,journal,comsoc]{IEEEtran}

\usepackage{amsmath,amsfonts}
\usepackage{amsthm}
\usepackage{array}
\usepackage{textcomp}
\usepackage{stfloats}
\usepackage{url}
\usepackage{verbatim}
\usepackage{graphicx}
\usepackage{cite}
\usepackage{amssymb}
\usepackage{color,xcolor}

\usepackage{algorithm}
\usepackage{algpseudocode}
\usepackage[nice]{nicefrac}
\usepackage{orcidlink}    
\usepackage{bm,bbm,mathtools}
\usepackage{subcaption}
\usepackage{multirow}

\newcommand\numberthis{\addtocounter{equation}{1}\tag{\theequation}}
\newtheorem{lemma}{Lemma}

\newtheorem{theorem}{Theorem}
\newcommand{\qmax}{Q_{\mathrm{max}}}
\newcommand{\nrb}{N_{\mathrm{RB}}}
\newcommand{\nre}{N_{\mathrm{RE}}}\newcommand{\dserv}{D_{\mathrm{s}}}
\newcommand{\dwait}{D_{\mathrm{w}}}

\newcommand{\dtot}{D}
\newcommand{\dvp}{\mathcal{D}}
\newcommand{\snr}{\gamma}
\newcommand{\snrInst}{{S}}
\newcommand{\hvar}{{\mu_{h^2}}}
\newcommand{\ddecod}{{\zeta}}
\definecolor{iG}{rgb}{0.0, 0.7, 0.0}

\begin{document}
\title{Delay Analysis of 5G HARQ in the Presence of Decoding and Feedback Latencies}
\author{Vishnu N Moothedath\orcidlink{0000-0002-2739-5060}, 
        Sangwon Seo\orcidlink{0000-0002-9181-9454},
        Neda Petreska\orcidlink{0000-0002-2743-5800},
        Bernhard Kloiber,
        James Gross\orcidlink{0000-0001-6682-6559} 
\thanks{Vishnu N Moothedath, Sangwon Seo, and James Gross are with the Department of Intelligent Systems, KTH Royal Institute of Technology, Stockholm, Sweden (e-mail: vnmo@kth.se, sangwon@kth.se, jamesgr@kth.se).}
\thanks{Neda Petreska and Bernhard Kloiber are with Siemens AG, Munich, Germany (e-mail: neda.petreska@siemens.com, bernhard.kloiber@siemens.com).}
}




\maketitle
\begin{abstract}
The growing demand for stringent quality of service (QoS) guarantees in 5G networks requires accurate characterisation of delay performance, often measured using Delay Violation Probability (DVP) for a given target delay. Widely used retransmission schemes like Automatic Repeat reQuest (ARQ) and Hybrid ARQ (HARQ) improve QoS through effective feedback, incremental redundancy (IR), and parallel retransmission processes.
However, existing works to quantify the DVP under these retransmission schemes overlook practical aspects such as decoding complexity, feedback delays, and the resulting need for multiple parallel ARQ/HARQ processes that enable packet transmissions without waiting for previous feedback, thus exploiting valuable transmission opportunities. This work proposes a comprehensive multi-server delay model for ARQ/HARQ that incorporates these aspects. Using a finite blocklength error model, we derive closed-form expressions and algorithms for accurate DVP evaluation under realistic 5G configurations aligned with 3GPP standards. Our numerical evaluations demonstrate notable improvements in DVP accuracy over the state-of-the-art, highlight the impact of parameter tuning and resource allocation, and reveal how DVP affects system throughput. 
\end{abstract}


\IEEEpeerreviewmaketitle

\section{Introduction}\label{sec:intro}
The advent of 5G networks has marked a significant transformation in wireless communication, for instance, by supporting ultra-reliable and low-latency communication (URLLC), enhanced mobile broadband (eMBB), and massive machine-type communications (mMTC) services~\cite{3gpp2017study,popovski20185G,3gpp.22.104}. 
URLLC demands arguably the strictest quality of service (QoS) requirements in 5G in terms of delay and reliability and is poised to be the main enabler for real-time applications such as autonomous driving, virtual reality, and Industry 4.0\cite{sisinni2018industrial}. 
These applications are typically characterised by short packets transmitted with moderately low throughput\cite{durisi2016URLLC} and require delays in milliseconds with very low packet error rates (PER) of at most $10^{-3}$~\cite{popovski20226g} to $10^{-5}$~\cite{3gpp.38.913}, between various devices like machines, sensors, actuators and controllers.

To meet such strict QoS requirements in 5G, increasing the coding capabilities and reducing the PER beyond a limit is neither feasible with the timing constraints nor cheap in terms of resource costs. 
It is not effective either, as the decoding complexity has a significant negative effect on fulfilling the QoS requirements~\cite{celebi2021latency}. It has been argued that for given channel conditions, it is suboptimal to aim solely to minimise the PER~\cite{Peng2011ReliabilityPHY}, but it is better to aim for a moderate error rate with a good retransmission mechanism.

Automatic repeat request (ARQ)~\cite{ArqCommEngDeskRef} and hybrid automatic repeat request (HARQ)~\cite{Frenger2001HarqHSDPA,DAHLMAN2014299} retransmission schemes have already become ubiquitous in wireless communication. They enhance reliability and reduce latency by effective feedback, selective retransmissions, and incremental redundancy (IR) in the case of HARQ. HARQ, for instance, significantly outperforms no-feedback schemes for low-latency targets under the assumption of limited frequency diversity and no time diversity~\cite{Ostman2018Harq}, typical of the short packet URLLC.
Further, to reduce the latency, these retransmission schemes are implemented as a multi-process transmit queue, where the packets do not wait for feedback from the previous packets. These parallel transmissions of unacknowledged packets are called \textit{ARQ/HARQ processes}.
The need for these ARQ/HARQ processes arises from the decoding complexity and feedback scheduling. This is because, with a single ARQ/HARQ process, the packets have to wait for feedback, wasting all the valuable transmission opportunities during the round-trip time (RTT) of the packet.

In a 5G system, a stricter target latency typically comes at the cost of reduced reliability.
Achieving a sweet spot in the reliability-latency trade-off is thus essential, which is generally measured using the delay violation probability (DVP) of a target delay. 
Current state-of-the-art methods for computing DVP rely on single-server approximations of multi-process ARQ/HARQ schemes, which provide accurate estimates only when the inherent decoding and feedback delays are neglected. 
In this work, we address these critical gaps and explore the relation between DVP and 5G retransmission schemes by modelling ARQ and HARQ as a multi-server queue in the presence of decoding complexity and feedback delay.

\subsection{Related Work}
Several studies have explored the delay performance of wireless networks with and without retransmissions.
From a queuing theoretic perspective, the trade-off between error probability and delay of multi-access systems over AWGN channel is analysed in~\cite{Telatar1995}, and analytical models are developed to compute end-to-end delay in wireless networks modelled as a G/G/1 queue in~\cite{bisnik2006queuing}.
Much work has been done to characterise and derive bounds on the performance of wireless networks using network calculus or large-deviation theory~\cite{devassy2019reliable,Zubaidi2016,yeh2012fundamental,petreska2019bound}.
Some of these include analysing delay and error performance using effective bandwidth~\cite{chang1995EB,hassan2004markov}, 
deriving delay bounds using effective capacity and service curve approaches~\cite{wu2003effectiveCap,fidler2006wlc15},
and deriving delay bounds and solutions for delay distributions using stochastic network calculus~\cite{jiang2008SNC,Ciucu2010} and (min,×) algebra~\cite{Zubaidi2016}.

The performance of ARQ and HARQ retransmission schemes has been widely studied in low-latency environments. 
An effective-capacity\cite{wu2003effectiveCap} analysis of general HARQ systems is given in Larsson et.al.~\cite{larsson2016HARQ}. However, this analysis relies on an asymptotic information-theoretic approach requiring large packets\cite{devassy2019reliable}.
Akin et al.~\cite{akin2015backlog} introduced a state transition model for HARQ systems and derived the effective capacity, modelling packet error rates using outage probability based on Shannon capacity~\cite{Shannon}. However, outage and ergodic capacity are more suited only for long packets and are not appropriate otherwise~\cite{durisi2016URLLC}. 
Further, Schiessl et al.~\cite{schiessl2015delay} analysed the delay of finite blocklength wireless fading channels and showed that the Shannon capacity model significantly overestimates the delay performance in low-latency applications.

To address this, The authors later studied the sensitivity of delay under the finite blocklength regime in \cite{schiessl2018delay} and derived an approximation for the decoding error probability under certain assumptions. 
Specifically on ARQ, Devassy et.al.~\cite{devassy2014finite,devassy2018delay,devassy2019reliable} used finite blocklength capacity over fading channels~\cite{polyanskiy2010FBL,polyanskiy2011feedback,Wei2013FBLblockfading} to study the performance of short packet communication.
In their work, they extended the concept of the slotted Gaussian collision channel with feedback~\cite{caire2000modulation,caire2001throughput} and studied the throughput and delay as a function of the coded packet size and HARQ as a special case. The authors showed the existence of significantly different DVP for the same average delay, thus cementing the fact that studies on average delay are not sufficient for providing useful QoS guarantees. Similar studies by Sahin et al.~\cite{Sahin2014Harq,Sahin2015Harq,Sahin2019Harq} focused on HARQ incremental redundancy (HARQ-IR)~\cite{DAHLMAN2014299} and analyzed its performance over Gilbert-Elliott channels with Rayleigh fading. They modelled HARQ as a Markov chain where the fading coefficients were discretized into states, with decoding errors modelled as outages on these discrete thresholds.

All the works are either restricted to a single-process retransmission scheme or model the multi-process ARQ/HARQ using a single-server queue. 
These limitations worsen the modelling inaccuracies for systems with larger RTTs, and fail to address practical implementation aspects of slot-based 5G systems, where inescapable decoding complexity and non-negligible feedback delays over multiple transmissions significantly contribute to the DVP.
While some works, such as~\cite{Sahin2019Harq}, include waiting delays in their analysis, they argue that cumulative transmission delays dominate the total delay. However, in slot-based 5G systems, even a single slot for decoding and feedback can constitute at least 50\% of the RTT, making this assumption less valid.
These studies are information-theoretic, lacking considerations for resource allocation and modulation and coding schemes (MCS), or are not sufficiently aligned with 3GPP specifications.
This limits their practical applicability, as real-world systems must account for the effects of resource allocation, coding schemes, and feedback delays on system performance.

\subsection{Contributions}
Our contributions are summarised as follows:
\begin{enumerate} 
\item We propose a framework consisting of a delay model and an error model to accurately compute the DVP for ARQ and HARQ retransmission schemes in 5G. This framework has the potential to aid resource allocation and link adaptation algorithms targeting specific DVPs at given delay thresholds.
\item The delay model employs multi-server transmit queues, enabling support for multiple ARQ/HARQ processes while accounting for decoding and feedback delays. The model is grounded in 3GPP standards and incorporates realistic configurations, providing a key advancement over existing works.
\item The error model, while simple, uses realistic finite blocklength theory to evaluate the PER of ARQ and HARQ-IR with sufficient accuracy. The error model is isolated from the delay model, enabling the results to be used with measured PER values, further increasing the model's flexibility for practical use.
\item Our numerical evaluations demonstrate accuracy improvements over single-server models that rely on immediate feedback (IF) assumptions.
We analyse the effect of parameter tuning on DVP across various delay regimes and target delays and highlight the impact of resource allocation and packet sizes, especially under tight delay constraints. 
Additionally, we reveal the existence of an optimal arrival rate that maximises the system throughput. 
\end{enumerate}

The remainder of this work is organised as follows: In Section~\ref{sec:model}, we introduce the system model and the error model. 
In Sections \ref{sec:arq} and \ref{sec:harq}, we propose closed-form expressions and algorithms to compute the DVP for ARQ and HARQ retransmission schemes. Within each of these sections, we (1) discuss the queuing model and compute the steady-state queue probabilities, (2) compute the wait delay distribution, and (3) compute the service delay distribution and use it to calculate the DVP. 
Finally, in Section~\ref{sec:simulations}, we show the numerical evaluation in detail and conclude in Section~\ref{sec:conclusion}.

\section{System Model and Problem Statement}\label{sec:model}
\begin{figure*}[th]
\centering
\begin{minipage}[t]{.5\textwidth}
  \centering
  \includegraphics[width=\textwidth]{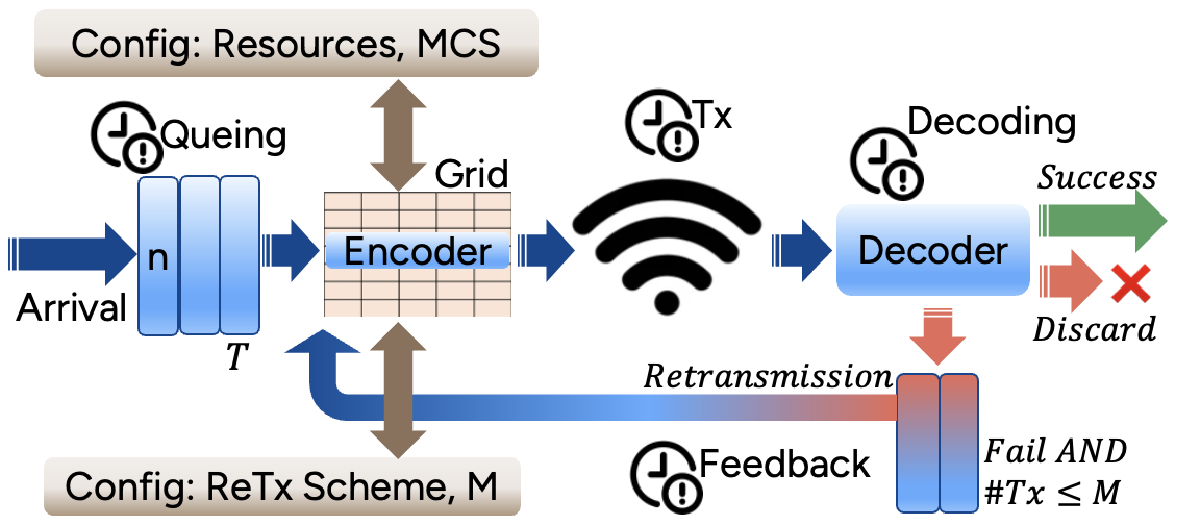}
  \caption{Illustration of the closed-loop communication system studied showing the retransmission process. Different delay components are shown where the packets experience them.}
  \label{fig:BasicModel}
\end{minipage}%
\hspace{0.03\textwidth}%
\begin{minipage}[t]{0.46\textwidth}
  \centering
  \includegraphics[width=\textwidth]{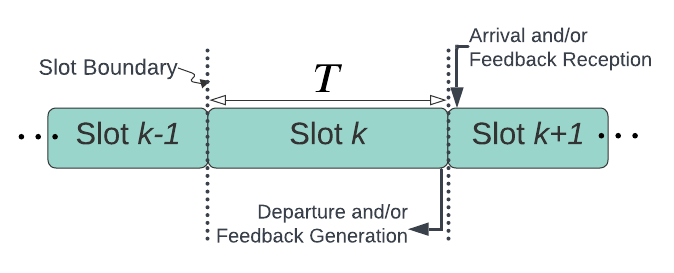}
  \caption{Timing diagram showing the order and positions of different slot-based arrival and departure events with respect to the corresponding slot boundary.}
  \label{fig:timingdiagram}
\end{minipage}
\end{figure*}

Consider a 5G communication system as depicted in Fig.~\ref{fig:BasicModel}. 
A User Equipment (UE) generates or receives uplink (UL) packets and queues them for transmission\footnote{The analysis and results apply to uplink and downlink scenarios; we focus on the uplink for clarity and consistency.}. 
These packets are sent to a dedicated gNB via a 5G-NR wireless link, utilizing a fixed number of resources scheduled to the UE in each time slot.
The packet is encoded over these pre-allocated frequency domain resource blocks (RB) using the configured modulation and coding scheme (MCS).

If the queue is non-empty, the UE uses the entire slot to transmit the head-of-the-queue packet. 
The gNB attempts to decode the packet using the implemented coding scheme. Successful packets are used for their intended purpose, and an acknowledgement is fed back in the downlink (DL)
The UE receives this feedback, and retransmission is triggered if necessary. 
For this, a static retransmission scheme is configured, and the packets are discarded after a maximum number of transmission attempts.

The process involves various delays, as shown in Fig.~\ref{fig:BasicModel}. 
Encoding delay is ignored, as it occurs only once per packet and can be performed while the packet waits in the queue. Similarly, propagation delays are neglected due to the short distances typical in high-reliability applications. 
However, if needed, one can include the encoding delay by subtracting it from the delay target, as it is endured only once per packet, and the propagation delay by adding it to the decoding delay, as they are always encountered as a pair.


\subsection{System Model}
Arrivals (or generation) of packets of size $n$ bits are modelled to occur randomly with an arrival probability of $f$ in each time slot.
While earlier arriving packets are given initial transmission opportunities, they do not wait for the feedback of previously transmitted packets, forming multiple simultaneous transmit-retransmit processes.
These packets form a queue awaiting transmission opportunities, which is modelled as a multi-server queue, with each server representing a packet undergoing a transmission process.
The maximum queue size is $\qmax$, and the slot length is $T$.

New or retransmitted packets are added to the queue at the UE immediately after the slot boundary. 
Each packet transmission uses all time domain resources within the slot, with departures or feedback generation occurring at the gNB just before the subsequent slot boundary, as illustrated in Fig.~\ref{fig:timingdiagram}. Although the slot timings at the UE and gNB may not align perfectly in terms of absolute clock time due to propagation delay, their offset remains constant. Synchronization of slot indices is maintained through the application of a timing advance (TA) computed during the initial handshake, ensuring that a transmission in slot $k$ of the UE is received within slot $k$ at the gNB.
In this work, retransmissions always have priority, and retransmission schedules are added to the head of the queue. 

We assume that packet failures in the UL manifest as decoding failures at the gNB, and negative acknowledgements (NACKs) are always successfully transmitted in the DL for these failed packets. The maximum number of retransmissions allowed is denoted by $M$ corresponding. Depending on whether $M$ is finite or infinite, we refer to this as a \textit{truncated} or \textit{persistent} retransmission scheme, respectively. 
The packet experiences a decoding delay of $\ddecod$ regardless of success, and the feedback incurs a delay of $\delta$ before being received by the UE, both measured in slots.
Transmitted packets are stored separately outside of the queue for potential retransmissions, preventing queue overflow even when $\qmax < \infty$. Therefore, a failed packet sent in slot $k$ is up for retransmission in slot $k + \ddecod + \delta + 1$.

The packet error rate (PER) is modelled in two ways. First, we consider the ARQ scheme, where failures are independent and identically distributed (i.i.d.) across different packets and transmission attempts. Second, we consider the HARQ scheme, where failures are identical only between packets but not between transmission attempts. ARQ discards information from failed transmissions and retries decoding independently, whereas HARQ retains this information to improve decoding of subsequent attempts. HARQ can be implemented in various ways, for example, Chase Combining (CC) and Incremental Redundancy (IR)\cite{DAHLMAN2014299}. In this work, we focus on HARQ-IR, the method that is predominantly used today. While the PER for ARQ is denoted by $p$, the PER for different transmission attempts in HARQ is represented by the vector $\bm{p} = [p_1, p_2, \dots, p_M]$, where $p_{m+1} \leq p_m,\,\forall m= 1, 2, \dots, M$. Thus, ARQ can be considered a special case of HARQ, where $p_m = p,\,\forall m$.

The slot-based packet transmission model above suffices for computing DVP given a known PER. However, to calculate PER and fully characterize DVP, we model transmissions based on a simplified 3GPP specification.
OFDM Resources are allocated in quanta of resource blocks (RBs), the number of which is denoted by $\nrb$. 
Each RB contains 12 sub-carriers separated in frequency with a sub-carrier spacing (SCS) of $15 \times 2^\nu$ kHz, indexed by $\nu$, referred to as numerology~\cite{3gpp.38.211}.
One OFDM sub-carrier defines a so-called resource element (RE), the number of which is denoted by $\nre$, resulting in $\nre = 12\nrb$.
A slot of duration $T = 2^{-\nu}$ ms contains 15 time-domain symbols, yielding $12\times15 = 180$ symbols per RB per slot and a blocklength of $180\times\nrb$ for each transmission.
In practice, only 12 or 14 symbols are present in a slot instead of 15 due to the cyclic prefix that we ignore in this work for simplicity. We also fix $\nu = 0$, setting SCS to 15 kHz. Nonetheless, these assumptions can be easily removed using the parameterised slot duration $T$ and symbols per slot. In addition, we assume a transport block size (TBS) not exceeding 8448 bits, avoiding code block segmentation~\cite{3gpp.38.212}.

Packets are transmitted over a Rayleigh block fading channel, with the assumption that the fading is constant within a time slot and changes independently between slots. Let $\mathcal{N}$ be the complex AWGN noise, and $h_k$ the Rayleigh-distributed fading coefficient at slot $k$ with $\mathbb{E}[\vert h_k\vert^2] = \hvar$. The channel input/output relation is:
$$y_k = h_k x_k + \mathcal{N}.$$
The transmission is carried out with a fixed MCS from the 3GPP 38.214-Table 5.1.3.1-1~\cite{3gpp.38.214}, which directly gives the spectral efficiency $\eta$ defined as the rate per symbol. Let $V$ denote the channel dispersion coefficient, $\snrInst$ the instantaneous SNR at the receiver, and $Q(x)$ the Q-function. The PER of an ARQ scheme with a given instantaneous SNR is given by\cite{polyanskiy2010FBL,Wei2013FBLblockfading}:
\begin{equation}
    p(\snrInst)=Q\left(\left(\log_2\left(1+\snrInst\right)-\eta\right)\sqrt{\frac{\nre}{V}}\right).\label{eq:secModel_perARQ_instSNR}
\end{equation}

HARQ-IR initially encodes with a low-rate code and generates a finite number $M$ of redundancy versions (RVs) through puncturing~\cite{HarqIRCommEngDeskRef}. 
Various methods exist for generating these RVs~\cite{HarqIRCommEngDeskRef,szczecinski2013rate,heidarzadeh2018systematic}, which differ in aspects such as RV overlap, whether the packet length changes with each RV and the number of higher layer packets combined in each physical layer packet. 
For simplicity, in this work, we assume that RVs are of equal length and non-overlapping, as illustrated in Fig.~\ref{fig:harqprocess2}, with 4 RVs corresponding to $M=4$.
Each RV has a coding rate of $Mr_c$ where $r_c$ is the coding rate of the unpunctured version. These are transmitted consecutively, thus completing the unpunctured code by the final attempt. 
Since HARQ uses previous transmissions for decoding, we get an effective spectral efficiency of $\nicefrac{\eta}{m}$ and a blocklength of $m\nre$ after $m^{\text{th}}$ transmission.
Using this, the PER for the $m^{\text{th}}$ transmission could be written as:
\begin{equation}
p_m(\vec{\snrInst})=Q\left(\left(\left(\frac{1}{m}\sum_{i=1}^m\log_2\left(1+\snrInst_{i}\right)\right)-\frac{\eta}{m}\right)\sqrt{\frac{m\nre}{V}}\right).\label{eq:secModel_perHARQ_instSNR}
\end{equation}
Here, $S_i$ is the SNR at the $i^{\text{th}}$ attempt, $\vec{\snrInst}=\{S_i\}$. We take the average capacity because HARQ decodes all the $m$ attempts jointly. This approach models a HARQ sufficiently well in terms of (only) the parameters of our interest while being much simpler than some existing approaches.
\begin{figure}[t]
\centering
\includegraphics[width=0.95\linewidth]{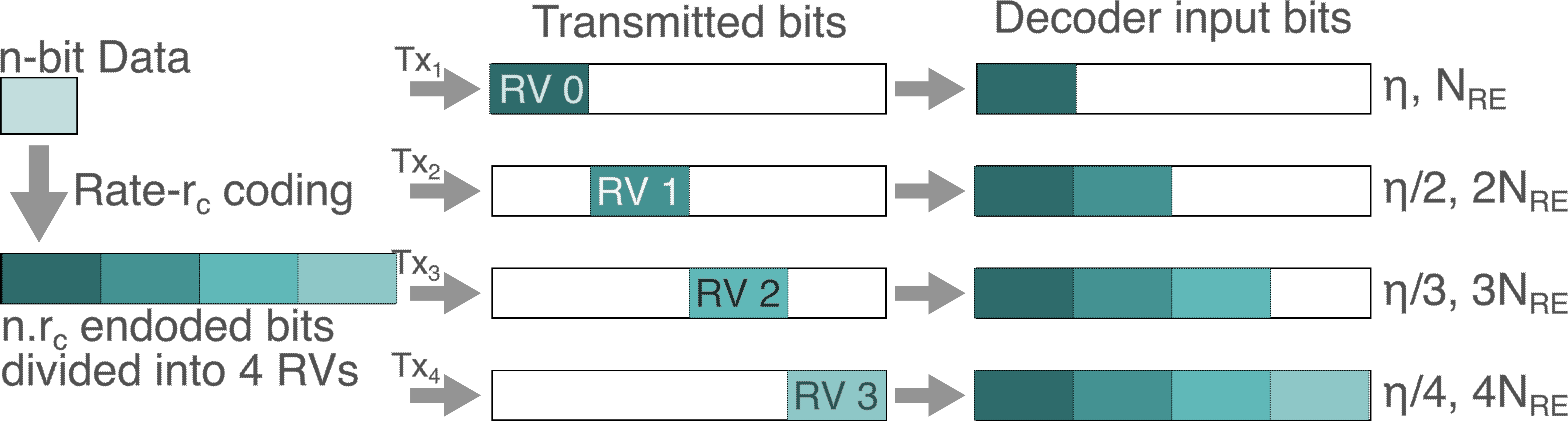}
\caption{Illustration of the HARQ-IR process. The $r_c$-rate channel coded bits are punctured to obtain 4 equal-length and non-overlapping RVs with a coding rate of $4r_c$ each. Effective spectral efficiency and block length after each decoding attempt are shown.}
\label{fig:harqprocess2}
\end{figure}

Recall that with a Rayleigh fading channel,  $\snrInst$ is exponentially distributed. Rewrite $S = \frac{\snr}{\hvar }\vert h\vert^2$, where $\snr \coloneqq \mathbb{E}[S]$, the average SNR at the receiver. To derive the PER for the average SNR $\snr$, one can compute the expectation over the distribution of the instantaneous SNR. 
For ARQ, we have,
\begin{equation}
p=\frac{1}{\hvar}\int_{0}^{\infty}p(\snrInst)\,e^{{-s}/{\hvar}}\,\mathrm{d}s.\label{eq:secModel_perARQ}
\end{equation}

Similarly, the PER $p_m$ for HARQ can be expressed as an $m$-dimensional integral using the joint distribution of $m$ i.i.d. SNR variables. While the joint PDF factorizes into a product, the integral remains inseparable due to the non-separability of the Q-function from \eqref{eq:secModel_perHARQ_instSNR}. 
\begin{equation}
p_m=\left(\frac{1}{\mu_h^2}\right)^m\int_{0}^{\infty}\dots\int_{0}^{\infty}p_m(\vec{\snrInst})\,\prod_{i=1}^{m}\left(e^{{-s_i}/{\mu_h^2 }}\,\mathrm{d}s_i\right)\label{eq:secModel_perHARQ}
\end{equation}

These expressions can be computed using numerical integration, Monte Carlo methods, or upper bounds on the Q-function. Though some bounds yield closed-form expressions, they are typically accurate only in limited parameter ranges and are unsuitable for varying resource allocation, packet sizes, and MCS choices. Therefore, in Section~\ref{sec:simulations}, we adopt a simple Monte Carlo approach. However, since our error model and queuing model are independent, any alternative method—analytical, numerical, or experimental—can be used to obtain the PER needed for the DVP computation.


\begin{table*}[t]
\centering
\renewcommand{\arraystretch}{1.25}
\caption{Table of abbreviations.}
\label{Tab:abbr}
\begin{tabular}{|ll|ll|ll|}
\hline
DVP & Delay Violation Probability    & IR  & Incremental Redundancy   & PER  & Packet Error Rate \\
MCS & Modulation and Coding Scheme & ARQ & Automatic Repeat Request & HARQ & Hybrid ARQ        \\
SCS & sub-carrier Spacing   & RE  & Resource Element         & RB   & Resource Block    \\ \hline
\end{tabular}
\end{table*}
\begin{table*}[t]
\centering
\renewcommand{\arraystretch}{1.25}
\caption{Table of notations.}
\label{Tab:notations}
\begin{tabular}{|ll|ll|ll|}
\hline
$n$       & Packet length (bits)     & $T$      & Slot duration (s)             & $f$                    & Frequency of random arrivals, with $f<1$   \\
$\qmax$   & Maximum queue size       & $M$      & Maximum transmission attempts & $c$                    & Cycle of deterministic arrivals (in slots) \\
$\nrb$    & Number of RBs            & $\nre$   & Number of REs                 & $k$                    & Slot index                                 \\
$\dwait$  & Waiting delay (in slots) & $\dserv$ & Serving delay (in slots)      & $\dtot$                & Total delay (in slots)                     \\
$d$       & Delay target (s)         & $\delta$ & Feedback delay (in slots)     & $\ddecod$              & Decoding delay (in slots)                  \\
$m$       & Transmission index       & $p$      & PER of the ARQ scheme         & $p_m$                  & PER of $m^{\text{th}}$ attempt of HARQ-IR  \\
$\dvp(d)$ & DVP for target $d$       & $\snr$   & Average SNR                   & $\hvar$                & Mean of $\vert h\vert^2$                   \\
$V$       & Channel dispersion       & $\eta$   & Spectral efficiency           & \multirow{2}{*}{$k_d$} & Maximum transmission attempts possible     \\
          &                          &          &                               &                        & without violating the target delay         \\ \hline
\end{tabular}
\end{table*}

\subsection{Problem Statement}
We consider three delay components: wait delay ($\dwait$), service delay ($\dserv$), and total delay ($\dtot$), all random variables measured in slots. Wait delay is the time between a packet’s arrival and its first transmission opportunity. Service delay is the time from the first transmission to the final transmission, and total delay is their sum:
\begin{equation*}
    \dtot = \dwait+\dserv.
\end{equation*}
As the feedback of the successful (or discarded) transmission is irrelevant, for $m$ transmission attempts:
\begin{equation}
    \dserv = m{+m\ddecod}+(m-1)\delta.\label{eq:SecModel_Dserv}
\end{equation}
The delay violation probability (DVP) associated with a delay target $d$ is defined as:
\begin{equation}
\dvp(d) = \mathbb{P}\left(D>d\right).\label{eq:SecModel_dvpDefinition}
\end{equation}
Our goal is to characterize the DVP as a function of various system parameters for ARQ and HARQ-IR. 

\subsubsection*{Warm-up: Bounded Arrival Retransmission Model}
For illustrative purposes, we consider a simplified ARQ scheme with no waiting time ($\dwait = 0$), resulting in $\dtot = \dserv$. 
This scheme, referred to as the Bounded Arrival Retransmission (BAR) scheme, assumes arrivals are either deterministic with a cycle $c\geq M \cdot \text{RTT}$ or are triggered only after the successful transmission of the previous packet, ensuring that a queue never forms. 
Here, the round-trip time (in slots) is given by $\text{RTT}=1+\ddecod+\delta$. 
Let $k_d$ denote the maximum number of transmission attempts allowable without violating the delay target $d$. Thus, the DVP corresponds to $k_d$ failed transmissions, that is, a probability of $p^{k_d}$. 
We have,
\begin{align}
    k_d &= \left\lfloor \frac{\nicefrac{d}{T} + \delta}{\delta + \ddecod + 1} \right\rfloor, \label{eq:secModel_Kd} \\
    \dvp(d) &= p^{k_d}. \label{eq:secModel_BARdvp}
\end{align}
It is worthwhile to observe that $k_d = \left\lfloor \nicefrac{d}{T} \right\rfloor$ when $\delta = \ddecod = 0$, where each attempt takes exactly 1 slot.\\

Note that with minimal effort, one can modify the general ARQ/HARQ results for a deterministic arrival process with a cycle time of every $c = f^{-1}$ time slots. We omit this part due to space constraints.
The results can be extended with minimal adjustments to a deterministic arrival process with a cycle time of $c = f^{-1}$ time slots. This extension is omitted due to space constraints.

We summarize important abbreviations in TABLE~\ref{Tab:abbr} and notations in TABLE~\ref{Tab:notations}.

\allowdisplaybreaks
\section{ARQ: Independent Retransmissions}\label{sec:arq}
Consider an ARQ retransmission scheme with independent failure events. The arrivals occur randomly with a probability $f$, forming a FIFO queue. Failed transmissions are added back to the head of the queue after $\delta$ slots, as described in Section~\ref{sec:model}. We first compute the wait delay and service delay separately and then derive the total delay. 

\subsection{Queing Model}
The UE buffer is modelled as a discrete-time Markov chain where the states represent the queue length, including the packet currently being served. State observations are made at the slot boundary. As mentioned in the system model, a departure occurs with every transmission attempt (with a probability of $1$). This is immediately followed, in order, by a retransmission schedule for failed packets, the slot boundary and the new arrivals. The retransmission is scheduled for a slot $\delta+1$ after the corresponding transmission slot. 

It is straightforward to see that the state transitions from state $q$ to $q+1$ corresponds to an arrival at the immediate next slot, say $k$, and a transmission failure at slot $k-\delta-1$. 
The transmission failure at slot $k-\delta-1$ is given by $p\left(1-\pi_0(1-f)\right)$ where $\pi_0$ denotes the probability of an empty queue. Here we take care of the fact that either the queue needs to be non-empty or there should be a fresh arrival for a transmission to happen in the first place to get a failed transmission.
Thus we obtain the Markov chain shown in Fig.\ref{fig:markovChain_modified}, where $\hat{p} = p\left(1-\pi_0(1-f)\right)$ and $\hat{p}' = 1-\hat{p}$. 
The residual self-loop probabilities of $1-f\hat{p}$ for state~0 and $f\hat{p}' + f'\hat{p}$ for all other states are not shown in the figure.

\begin{figure*}[t]
\centering
\begin{tikzpicture}
    \node[draw, circle,minimum size=1cm] (s0) at (0,0) {0};
    \node[draw, circle,minimum size=1cm] (s1) at (3,0) {1};
    \node[draw, circle,minimum size=1cm] (s2) at (6,0) {2};
    \node[draw, circle,minimum size=1cm] (s3) at (9,0) {3};
    
    \node at (10.5,0) {\Huge\dots};
    
    \draw[->] (s0) to[bend left] node[midway, above] {$f\hat{p}$} (s1);
    \draw[->] (s1) to[bend left] node[midway, above] {$f\hat{p}$} (s2);
    \draw[->] (s2) to[bend left] node[midway, above] {$f\hat{p}$} (s3);

    \draw[->] (s1) to[bend left] node[midway, below] {$f'\hat{p}'$} (s0);
    \draw[->] (s2) to[bend left] node[midway, below] {$f'\hat{p}'$} (s1);
    \draw[->] (s3) to[bend left] node[midway, below] {$f'\hat{p}'$} (s2);

    \draw[->] (s0) edge[loop left] node[left] {} (s0);
    \draw[->] (s1) edge[loop above] node[above] {} (s1);
    \draw[->] (s2) edge[loop above] node[above] {} (s2);
    \draw[->] (s3) edge[loop above] node[right] {} (s3);
\end{tikzpicture}
\caption{Markov chain with queue length as states as observed at the slot boundary for an ARQ scheme. The transition probabilities, except the self-loop probabilities, are shown.}
\label{fig:markovChain_modified}
\end{figure*}
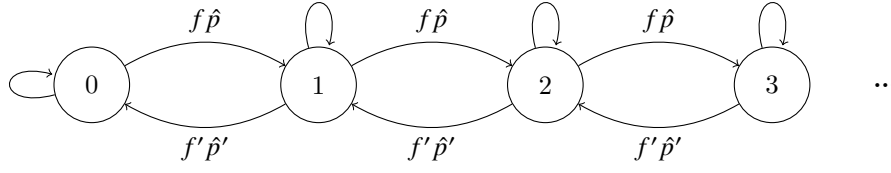

\subsubsection*{Steady state probabilities}
We now focus on determining the steady-state probabilities (SSP) of the queue. Rather than directly finding the SSP of the initial Markov chain, we bound it using the SSP of a modified Markov chain representing an \textit{immediate feedback scenario} with $\delta = \ddecod = 0$, which is mathematically more tractable.

In such a scenario, the retransmission happens in the immediate next slot. One can alternatively consider that the departure occurs with a probability of $1-p$ instead of $1$, and there is no retransmission scheduled. 
The events between two state observations are an arrival at the start of the slot and a departure at the end of the slot with probabilities $f$ and $1-p$, respectively. 
Thus, the transition from state $q$ to $q+1$ comes with a probability $fp$, larger than $f\hat{p}$. 
Therefore, the CCDF of the queue length of this adapted Markov chain stochastically dominates that of the Markov chain from \ref{fig:markovChain_modified}. We will elaborate on this soon and use it to bind the violation probability of the wait delay. 

Let $\pi_i$ denote the steady-state probability of the adapted Markov chain. Lemma~\ref{lemma_arq_queue} provides the CCDF of the queue length.

\begin{lemma}\label{lemma_arq_queue}
    The CCDF of the queue length $Q$ is given by:
    \begin{equation}
        \mathbb{P}\left(Q>q\right) = \left(\dfrac{fp}{(1-f)(1-p)}\right)^{q+1}\label{eq:secARQ_QueueCCDF}
    \end{equation}
\end{lemma}
\begin{proof}
 \begin{align*}
    \pi_0 &= (1-fp)\pi_0 + f'p'\pi_1, \\
    \Rightarrow \pi_1 &= \dfrac{fp}{f'p'}\pi_0.\\
    \pi_1 &= fp\pi_0+ (1-fp-f'p')\pi_1+ f'p'\pi_2, \\
    \Rightarrow \pi_2 &= \dfrac{fp}{f'p'}\pi_1,\\
     &=  \left(\dfrac{fp}{f'p'}\right)^2\pi_0.
\end{align*}
Continuing with the same steps, we get
\begin{align*}
    \!\pi_i &= \left(\dfrac{fp}{f'p'}\right)^{i}\pi_0,\,\forall i\geq0,
    \intertext{that is,}
\!\pi_i &= \left(\dfrac{fp}{(1-f)(1-p)}\right)^{i}\!\!\left(1-\dfrac{fp}{(1-f)(1-p)}\right)\!,\forall i\geq0.\numberthis\label{eq:secARQ_ssp}
\end{align*}
Let $Q$ be the random variable denoting the queue length in the immediate feedback scenario, the distribution of which is given in \eqref{eq:secARQ_ssp}. Thus,
\begin{align*}
    \mathbb{P}\left(Q>q\right)&=\left(1-\dfrac{fp}{(1-f)(1-p)}\right)\\
    &\qquad\qquad\quad\sum_{i=q+1}^{\infty}\left(\dfrac{fp}{(1-f)(1-p)}\right)^{i},\,\forall i\geq0.\\
    &=\left(\dfrac{fp}{(1-f)(1-p)}\right)^{q+1}.
\end{align*}   
\end{proof}

Observe that the queue is stable if the arrival rate does not exceed the departure rate, i.e. if $f \leq 1 - p$ or equivalently if $\frac{p}{1-f} \leq 1$. This can also be derived by computing the expected queue length $\Bar{\pi}$:
\begin{align*}
    \Bar{\pi} &=\left(1-\dfrac{fp}{(1-f)(1-p)}\right)\sum_{i=0}^{\infty}i\left(\dfrac{fp}{(1-f)(1-p)}\right)^{i},\\
    &=\left(1-\dfrac{fp}{(1-f)(1-p)}\right)\frac{(1-f)(1-p)}{(1-f-p)^2}fp,\\
    &=\frac{fp}{1-f-p},
\end{align*}
which implies stability when:
\begin{equation}
    \frac{p}{1-f}\leq1.\numberthis\label{eq:secARQ_zerofdbk_stability}
\end{equation}

\subsection{Wait delay}
It is clear from \eqref{eq:secARQ_QueueCCDF} that the queue length distribution of the immediate feedback scenario with PER $p$ stochastically dominates the queue length distribution of a delayed feedback scenario with a PER $\hat{p} < p$ in a first-order stochastic dominance sense~\cite{whang2019econometric}.

Let $\dwait\vert Q$ be the wait delay, conditioned on the queue length. Thus, 
\begin{align*}
    \mathbb{P}(\dwait =k)&=\sum_q\,\pi_q\mathbb{P}(\dwait ={k}\vert Q=q).
\end{align*}
As the wait delay increases with queue size, the stochastic dominance of the queue length of the delayed feedback scenario also implies the stochastic dominance of the corresponding wait delay. This will also become evident from Lemma~\ref{lemma_arq_wait}, where the upper bound, which is the CCDF of the wait delay with immediate feedback, decreases as the PER decreases. 
This upper bound is found to be sufficiently tight from simulations. This is because, unlike the service delay, the wait delay measured from arrival to the first transmission is largely unaffected by the feedback.

Recall that the sum of i.i.d. geometric random variables follows a negative binomial distribution\cite{johnson2005univariate}. 
For $X_{q}$ representing such a sum, the probability mass function is given by:
\begin{equation}
    \mathbb{P}(X_{q}=k)=\dfrac{(k-1)!}{(q-1)!(k-q)!}(1-p)^qp^{k-q}.\label{eq:secARQ_geoSum}
\end{equation}
\begin{lemma}\label{lemma_arq_wait}
    The CCDF of the wait delay is given by:
    \begin{equation}
  \mathbb{P}(\dwait >j) \leq\frac{f}{1-p}\left(\frac{p}{1-f}\right)^{j+1}.\label{eq:secARQ_waitdelay} 
    \end{equation}
\end{lemma}
\begin{proof}
For the immediate feedback scenario, the number of transmissions attempted by a packet in the queue is distributed geometrically. Thus, $\mathbb{P}(D_w\vert Q)$ is given by the sum of $Q$ iid geometrically distributed random variables. We have,
\begin{align*}
\mathbb{P}(\dwait\!>\!j)&\leq 1\!-\!\left(\pi_0+\sum_{k=1}^{j}\sum_{q=1}^{k}\pi_q\,\mathbb{P}(D_w=k\vert Q=q)\right),\,\forall j\geq0.
\end{align*}
 Here, $\pi_0$ represents an empty queue. Let $Z(k)$ denote the inner sum. Expanding with \eqref{eq:secARQ_geoSum}, we have:
 \begin{equation}
Z(k)=\sum_{q=1}^{k}\pi_q\dfrac{(k-1)!}{(q-1)!(k-q)!}(1-p)^qp^{k-q},\,\forall k\geq1.
 \end{equation}
\begin{align*}
&=\sum_{q=1}^{k} 
\left( \frac{f \, p}{(1 - f) \, (1 - p)} \right)^q 
\left(1 - \frac{f \, p}{(1 - f) \, (1 - p)} \right)\\&\hspace{8em}
\left( \frac{(k - 1)!}{(k - q)!\,(q - 1)!} \right)
(1 - p)^q \, p^{k - q},\\
&=\left(1 - \frac{f \, p}{(1 - f) \, (1 - p)} \right)\sum_{q=0}^{k-1} 
\left( \frac{f \, p}{(1 - f) \, (1 - p)} \right)^{q+1} \\&\hspace{6em}
\left( \frac{(k - 1)!}{(k - q-1)!\,(q)!} \right)
(1 - p)^{q+1} \, p^{k - q-1},\\
&=\left(1 - \frac{f \, p}{(1 - f) \, (1 - p)} \right)
p^{k}
\sum_{q=0}^{k-1} 
\left( \frac{f}{1 - f} \right)^{q+1}\\&\hspace{13em}
\left( \frac{(k - 1)!}{(k - 1 -q)!\,(q)!} \right),\\
&=\left(1 - \frac{f \, p}{(1 - f) \, (1 - p)} \right)
p^{k}\left( \frac{f}{1 - f} \right)
\sum_{q=0}^{k-1} {\binom{k-1}{q}}
\left( \frac{f}{1 - f} \right)^{q},\\
&=\left(1 - \frac{f \, p}{(1 - f) \, (1 - p)} \right)
p^{k}\left( \frac{f}{1 - f} \right)
\left(1+ \frac{f}{1 - f} \right)^{k-1},\\
&=f\,\left( \frac{p}{1 - f} \right)^{k}
\frac{1-f-p}{(1 - f) \, (1 - p)}.
\end{align*}
\begin{equation*}
 \Rightarrow\mathbb{P}(\dwait >j) \leq 1-\left(\pi_0+\sum_{k=1}^{j}Z(k)\right),  
\end{equation*}
\begin{align*}
&=1-\left(\pi_0+\sum_{k=1}^{j}f\,\left( \frac{p}{1 - f} \right)^{k}
\frac{1-f-p}{(1 - f) \, (1 - p)}\right),\\
 &= 1-\left(\frac{1-f-p}{(1 - f) \, (1 - p)}\right)\left(1 + f\sum_{k=1}^{j}\left( \frac{p}{1 - f} \right)^{k}\right),\\
 &= 1-\left(\frac{1-f-p}{(1 - f) \, (1 - p)}\right)\\&\hspace{6em}\left(1 + \frac{fp}{1-f-p}\left(1-\left(\frac{p}{1-f}\right)^j\right)\right),\\
&= 1-\left(1-f-p+fp-fp\left(\frac{p}{1-f}\right)^j\right)\\&\hspace{13.5em}\left(\frac{1}{(1 - f) \, (1 - p)}\right),\\
&=\frac{fp}{(1-f)\,(1-p)}\left(\frac{p}{1-f}\right)^j,\\
&=\frac{f}{1-p}\left(\frac{p}{1-f}\right)^{j+1}.
\end{align*}
\begin{equation}
    \Rightarrow\mathbb{P}(\dwait >j) \leq\frac{f}{1-p}\left(\frac{p}{1-f}\right)^{j+1}.\nonumber
\end{equation}
\end{proof}

\subsection{Delay Violation Probability}\label{NoIR_dvp}
To get the DVP, we combine the upper bound of $\dwait$ with the service delay. The service delay is geometrically distributed based on the failure probability at the corresponding attempt, with values depending on $\delta$. It is given by:
\begin{align}
    \mathbb{P}\left(\dserv =k(\ddecod+1)+\delta(k-1)\right) &= p^{k-1}(1-p).\label{eq:secARQ_servDelay}\\
    \mathbb{P}\left(\dserv >k(\ddecod+1)+\delta(k-1)\right) &= p^{k}.\nonumber
\end{align}
Recall $k_d$ defined as the maximum number of transmissions possible before the service delay alone exceeds the delay target:
\begin{align}
    k_d\,\coloneqq& \underset{i}{Max}\,\left\{\,i:i(\ddecod+1)+(i-1)\delta\leq\left\lfloor\frac{d}{T}\right\rfloor\,\right\}\nonumber\\
    =&\left\lfloor\frac{\frac{d}{T}+\delta}{\delta+\ddecod+1}\right\rfloor.\label{eq:secARQ_kd}
\end{align}
\begin{theorem}\label{theorem_arq_dvp}
The DVP of the ARQ scenario for a delay target $d$ is given by:
\begin{multline}
\mathbb{P}(\dtot >d)\leq p^{ k_d } \\+  
\dfrac{f\,\left(\frac{p}{1-f}\right)^{  \left\lfloor\nicefrac{d}{T}\right\rfloor +\delta}
\left(1-p^{ k_d }\left(\frac{p}{1-f}\right)^{- k_d (1+\delta+\ddecod)}
\right)}
{f+\left(\frac{p}{1-f}\right)^{\delta+\ddecod}-1}.\numberthis\label{dvpbound}
\end{multline}
\end{theorem}
\begin{proof}
We have,
\begin{align*}
\mathbb{P}(\dtot >d)&\leq \sum_{i}\mathbb{P}(\dserv =i)\mathbb{P}(\dwait >d-i),\\
&= p^{k_d} + \sum_{i=1}^{{ k_d }} (1 - p) p^{i - 1}
\\&\hspace{2.5em}\left( \frac{f}{1-p}\left(\frac{p}{1-f}\right)^{1 + \left\lfloor\nicefrac{d}{T}\right\rfloor - \big(i(\ddecod+1) + (i - 1) {\delta}\big)}\right).
\end{align*}
The first term corresponds to the probability that the service delay alone exceeds the target. The second term accounts for a successful transmission at the $i^{\mathrm{th}}$ attempt with a service delay of $i(\ddecod+1) + (i-1)\delta$, along with all possible wait delays from \eqref{eq:secARQ_waitdelay} that in combination violate the delay target.
\begin{multline}
    \Rightarrow\mathbb{P}(\dtot >d) \leq p^{k_d} + 
f\left(\frac{p}{1-f}\right)^{1 + \left\lfloor\nicefrac{d}{T}\right\rfloor + \delta}\,
\\\sum_{i=1}^{{ k_d }}p^{i - 1} \left(\frac{p}{1-f}\right)^{-i(1+\delta+\ddecod)}
\end{multline}
\begin{align*}
&= p^{k_d} + 
f\left(\frac{p}{1-f}\right)^{1 + \left\lfloor\nicefrac{d}{T}\right\rfloor + \delta}\,
\\&\hspace{10.5em}\sum_{i=0}^{{ k_d -1 }} p^{i} \left(\frac{p}{1-f}\right)^{-(i+1)(1+\delta+\ddecod)}\\
&= p^{k_d} + 
f\left(\frac{p}{1-f}\right)^{1 + \left\lfloor\nicefrac{d}{T}\right\rfloor + \delta}\,
\left(\frac{p}{1-f}\right)^{-(1+\delta+\ddecod)}\,
\\&\hspace{10.5em}\sum_{i=0}^{{ k_d -1 }} p^{i} \left(\frac{p}{1-f}\right)^{-i(1+\delta+\ddecod)}\\
&= p^{k_d} + 
f\left(\frac{p}{1-f}\right)^{ \left\lfloor\nicefrac{d}{T}\right\rfloor-\ddecod}\,
\sum_{i=0}^{{ k_d -1 }} p^{i} \left(\frac{p}{1-f}\right)^{-i(1+\delta+\ddecod)}\\
&= p^{k_d} + 
f\left(\frac{p}{1-f}\right)^{ \left\lfloor\nicefrac{d}{T}\right\rfloor-\ddecod}\,
\frac{1-p^{k_d}\left(\frac{p}{1-f}\right)^{-k_d(1+\delta+\ddecod)}}{1-p\left(\frac{p}{1-f}\right)^{-(1+\delta+\ddecod)}}
\end{align*}
\begin{multline}
    \Rightarrow\mathbb{P}(\dtot >d)\leq p^{ k_d } \\-  
\dfrac{f\,\left(\frac{p}{1-f}\right)^{  \left\lfloor\nicefrac{d}{T}\right\rfloor +\delta}
\left(1-p^{ k_d }\left(\frac{p}{1-f}\right)^{- k_d (1+\delta+\ddecod)}
\right)}
{1-f-\left(\frac{p}{1-f}\right)^{\delta+\ddecod}}
\end{multline}
\end{proof}

\section{HARQ: Incremental Redundancy}\label{sec:harq}
In this section, we consider the HARQ scenario with incremental redundancy (IR). 
As discussed earlier, we assume that the coded packet length for all transmissions remains constant, thereby attaining the maximum increment in redundancy with reach retransmission. 
The PER is represented by the vector $\Vec{p}=\begin{bmatrix}p_1 & p_2 & \dots & p_M\end{bmatrix}$. We assume a maximum of $M$ transmissions and a maximum of $Q_{\mathrm{max}}$ parallel HARQ processes. Typically, $M=4$, and $Q_{\mathrm{max}}$ is 8 or 16 in real HARQ implementations.
\begin{table}[b]
\centering
\renewcommand{\arraystretch}{1.25}
\renewcommand{\arraystretch}{1.25}
\begin{tabular}{|l|l|l|l|l|}
\hline
\multicolumn{1}{|c|}{\multirow{2}{*}{State}} & \multicolumn{1}{c|}{\multirow{2}{*}{Next state}} & \multicolumn{1}{c|}{Range} & \multicolumn{1}{c|}{Range} & \multicolumn{1}{c|}{Proba-} \\
\multicolumn{1}{|c|}{}                       & \multicolumn{1}{c|}{}                            & \multicolumn{1}{c|}{$(q)$}   & \multicolumn{1}{c|}{$(m)$}   & \multicolumn{1}{c|}{bility} \\ \hline
$(0,1)$                                      & $(0,1)$                                          &   -                        &   -                        & $1-fp_1$                    \\
$(0,1)$                                      & $(1,2)$                                          &   -                        &   -                        & $fp_1$                      \\
$(q,m)$                                      & $(q,m+1)$                                        & $[1,\qmax)$                & $[1,M)$                    & $f'p_m$                     \\
$(q,m)$                                      & $(q+1,m+1)$                                      & $[1,\qmax)$                & $[1,M)$                    & $fp_m$                      \\
$(q,m)$                                      & $(q,1)$                                          & $[1,\qmax]$                & $[1,M)$                    & $fp_m'$                     \\
$(q,m)$                                      & $(q-1,1)$                                        & $[1,\qmax]$                & $[1,M)$                    & $f'p_m'$                    \\
$(q,M)$                                      & $(q,1)$                                          & $[1,\qmax]$                &    -                       & $f$                         \\
$(q,M)$                                      & $(q-1,1)$                                        & $[1,\qmax]$                &    -                       & $f'$                        \\
$(\qmax ,m)$                                 & $(\qmax ,m+1)$                                   &        -                   & $[1,M)$                    & $p_m$                       \\ \hline
\end{tabular}
\caption{Non-zero probabilities of the transition probability matrix $P$ for a given $q$ and $m$. State number $s=qM+m$ for the state $(q,m)$. $f'=1-f,\;p_m'=1-p_m.$}
\label{table:IR}
\end{table}

To compute the DVP, we proceed similarly to the previous section by combining the wait delay and service delay, which are computed separately. 
We propose an algorithmic approach to compute the wait delay, as this is more suited for HARQ with a relatively small $M$, $\qmax$, and a non-iid PER across the retransmissions.
As before in Section~\ref{sec:arq}, we bound the wait delay using the immediate feedback scenario where the retransmissions happen in the immediate next slot.
We now construct the Markov chain transition probability matrix of this scenario.

\subsection{Queueing Model}
We define $(\qmax+1)M$ states, denoted by the tuple $(q,m)$, where $0 \leq q \leq Q_{\mathrm{max}}$ and $1 \leq m \leq M$, measured at the slot boundary. The states represent the current queue length $q$ (observed by a newly arriving packet) and the transmission number $m$ of the packet that will be transmitted in the next slot. 
For example, the state $(3,2)$ indicates that the queue length of $3$ and the packet to be transmitted has already failed once. 

The non-zero transition probabilities for all states are given in Table~\ref{table:IR}. For ease in constructing and using the transition probability matrix $P_{(q+1)M\times (q+1)M}$, we number the states as $s=1,2,\dots,(qM+m),\dots,(q+1)M$. The states $(0,m), m \geq 2$ are never reached and are included for uniformity and simplicity. These states are defined with a self-loop probability of 1 and have a steady-state probability of 0.


The probabilities can be explained as follows: $f$ represents arrival, and $p_m$ represents the PER at the $m^{\mathrm{th}}$ attempt. For state $(0,1)$, an arrival and transmission failure lead to a transition to state $(1,2)$, while other possibilities result in a loop. In states with $m=M$, the PER becomes irrelevant because the packet is either successfully transmitted or discarded. Similarly, a packet is dropped when an arrival and transmission failure occurs at state $(Q_{\mathrm{max}},m)$ due to a queue overflow. For other states, the transitions follow the typical pattern: failures increase $m$, successes reset $m$ to 0, and arrivals/departures adjust the queue size based on the transmission outcome.

Once $P$ is constructed, the steady-state probabilities, denoted by $\Tilde{\pi}$, can be computed by finding the eigenvector of $P^T$ corresponding to the unit eigenvalue. This can be done using standard algorithms or by iterating $P$ until $P^i \approx P^{i+1}$, with the rows converging to the steady-state probabilities.

The steady-state probabilities $\Tilde{\pi}$ are for the modified Markov chain with $(Q_{\mathrm{max}}+1)M$ states. To obtain the steady-state probabilities ${\pi}$ for each queue length $q=1,2,\dots, Q_{\mathrm{max}}$, we sum the probabilities of all states with the same queue length but different $m$ values:
\begin{equation}
    \pi_q = \sum_{s=qM}^{(q+1)M}\Tilde{\pi}_s.
\end{equation}
We assume that $\qmax$ is chosen such that packet drops due to queue overflow are negligible, typical in a high-reliability setting. Otherwise, one could repeat with a larger $\qmax$. That being said, we do consider the drops emerging from packets reaching the retransmission limit of HARQ ($m=M$), which cannot be neglected.

\subsection{Wait Delay}
To compute the wait delay, we start by finding $f_W(k|q)$, the conditional wait probability given queue length $q$. We propose ALGORITHM~\ref{Alg:getWaitProbability} to compute this for a given $k,q,\Vec{p}$ and $M$ using combinatorics.
The unconditional wait delay pmf is obtained by marginalizing the queue length probabilities:
\begin{equation}
    f_{\dwait}(k)=\mathbb{P}(\dwait=k) \leq \sum_{q=0}^{\infty}\pi_qf_W(k|q).\label{eq:secIR_waitDelay}
\end{equation}
\begin{algorithm}[t]
\caption{Recursive function \textit{getWaitProbability} to compute the conditional probability of wait delay of $k$ slots given a queuelength of $q$ packets. The global constant $M_0 = M$ in the first call of the recursion.}
\begin{algorithmic}[th]
\Function{getWaitProbability}{$k, q, \Vec{p}, M, M_0$}
    \If{$k == q$}
        \State\Return $(1 - \Vec{p}_1)^q$
        \Comment{$k=q\Rightarrow$ all success.}
    \ElsIf{$k < q$ \textbf{or} $k > M \cdot q$}
        \State\Return $0$
        \Comment{Out of range, $prob = 0$.}
    \EndIf
    \State $prob \gets 0$
    \State $N \gets \min(\text{floor}((k - q) / (M - 1)), q)$\
        \\\Comment{Max \#packets with max attempts $=M$.}
    \For{$n = 0$ \textbf{to} $N$}
        \State $numSeqs \gets \dbinom{q}{n}$
        \State $seqProbFail \gets \prod_{i=1}^{M-1}\Vec{p}_{i}$
        
        \If{$M == M_0$}
            \State $seqProbSucc \gets 1$
            \\\Comment{Handle discard case when $M = M_0$.}
        \Else
             \State $seqProbSucc \gets (1 - \Vec{p}_M)$
        \EndIf
        
        \State $seqProb\gets(seqProbFail \cdot seqProbSucc)^n$
        \State $subSeqProb \gets$ \Call{getWaitProbability}{}
        \State \qquad\qquad\qquad\qquad($k-Mn, q - n, \Vec{p}, M - 1, M_0$)
        \\\Comment{Recursion.}
        \State $prob \gets prob + numSeqs\!\cdot\!seqProb\!\cdot\!subSeqProb$
    \EndFor
    \State \Return $prob$
\EndFunction
\end{algorithmic}
\label{Alg:getWaitProbability}
\end{algorithm}

\subsection{Delay Violation Probability}
We now compute the distributions of the service delay, similar to Section~\ref{NoIR_dvp}. 
The service delay is determined by the PER vector and $k_d$, the maximum number of transmissions allowed before exceeding the delay target. 
Unlike \eqref{eq:secModel_Kd}, where we assumed infinite retransmissions, here we limit $k_d$ by $M$:
\begin{align}
k_{d}&=\min\left(M,\left\lfloor\frac{\nicefrac{d}{T}+\delta}{\delta+\ddecod+1}\right\rfloor\right),\label{eq:secIR_kd}\\
\mathbb{P}(\dserv>d)&=\prod_{i=1}^{k_{d}}p_i.\label{eq:secIR_servDelay}
\end{align}
Let $k_{d-kT}$ denote the $k_d$ for the delay target $d-kT$. 
\begin{align*}
    k_{d-kT}&=\min\left(M,\left\lfloor\frac{\nicefrac{(d-kT)}{T}+\delta}{\delta+{\color{red}\ddecod}+1}\right\rfloor\right),\nonumber\\
    &=\min\left(M,\left\lfloor\frac{\nicefrac{d}{T}-k+\delta}{\delta+\ddecod+1}\right\rfloor\right).
\end{align*}

Finally, the total DVP is computed as before in Section~\ref{NoIR_dvp}, using the wait delay and service delay violation probabilities:
\begin{align}
    \mathbb{P}(\dtot>d) &=\sum_k \mathbb{P}(\dwait=k)\mathbb{P}(\dserv>d-k)\nonumber\\
    &\leq\sum_kf_{\dwait}(k)\prod_{i=1}^{k_{d-kT}}p_i.\label{eq:secIR_totDelay}
\end{align}

\section{Numerical Evaluation}\label{sec:simulations}
We begin this section on numerical evaluation by detailing the parameter configuration, including default settings, MCS selection processes, and PER computation methods. 
We then compare the proposed ARQ and HARQ DVP evaluation schemes with the state-of-the-art IF approximation, showing the importance of not ignoring the decoding and feedback delay. 
Following this, we study key DVP trends across varying system parameters by examining the impact of RTT, resource allocation, and arrival rate on the evaluated DVP. Throughout this section, we consider a persistent ARQ with unlimited retransmissions and queue size, i.e., $M = \qmax = \infty$ and a typical HARQ configuration of $M = 4$ and $\qmax = 16$. 

\subsubsection*{Parameter Configuration}
We proposed DVP evaluation for ARQ and HARQ across various system parameters, leading to numerous permutations of parameter settings and illustrations. However, for clarity and conciseness, we limit our evaluation to key configurations.
Since incorporating decoding and feedback delays and supporting parallel ARQ/HARQ processes is the key novelty of this work, we choose RTT and the delay threshold $d$, which directly influences the DVP. 
We consider allocated $\nrb$ and packet length to highlight resource allocation implications and the arrival frequency $f$ to study system throughput. 
Default parameter values and ranges are listed in TABLE~\ref{table:params}. We set $\hvar=1$, $\gamma = 10$ dB, and $V=1$ as $\vert V-1\vert<0.0414,\,\forall\,\text{SNR}>4.1$ dB in an AWGN channel~\cite{polyanskiy2010FBL}. In each figure, a subset of parameters is varied, and the default values are used for the rest.

While some parameters like $\nrb$ are configurable, others are application-specific or come from the device capabilities. For example, \cite{3gpp.22.261} outlines KPI requirements for various applications. Similarly, RTT depends on factors such as decoding capabilities, feedback scheduling delays, and priority levels assigned by the gNB.
To ensure broad applicability, we use parameters within a typical range and present results independent of specific applications or scenarios.
\begin{table}[t]
\centering
\renewcommand{\arraystretch}{1.25}
\begin{tabular}{|l|l|l|}
\hline
Parameter             & Default value  & Range                                                                                                                            \\ \hline
$n$                   & 100$\times$8 b & $\{30, 50, 100, 200\}\times8$ b                                                                                                  \\
$\eta_{\mathrm{min}}$ & 0.2344         & -                                                                                                                                \\
$\eta_{\mathrm{max}}$ & 5.5547         & -                                                                                                                                \\
$\nrb$                & 10             & $\left\{\left\lceil\frac{n}{180\eta_{\mathrm{max}}}\right\rceil,\left\lceil\frac{n}{180\eta_{\mathrm{min}}}\right\rceil\right\}$ \\
$\snr$                & 10 dB          & -                                                                                                              \\
$\ddecod$             & 1 slot         & -                                                                                                                                \\
$\delta$              & 2 slots        & -                                                                                                                                \\
RTT                   & 4 slots        & $[1,7]$ slots                                                                                                                    \\
$d$                   & 8.5 ms         & $[2, 20]$ ms                                                                                                                     \\
$f$                   & $\frac{1}{3}$  & -                                                                                                                                \\
$\hvar$               & 1              & -                                                                                                                                \\
$V$                   & 1              & -                                                                                                                                \\ \hline
\end{tabular}
\caption{Default values and range of important parameters.}
\label{table:params}
\end{table}
\begin{figure}[b]
\centering
\includegraphics[width=0.98\linewidth]{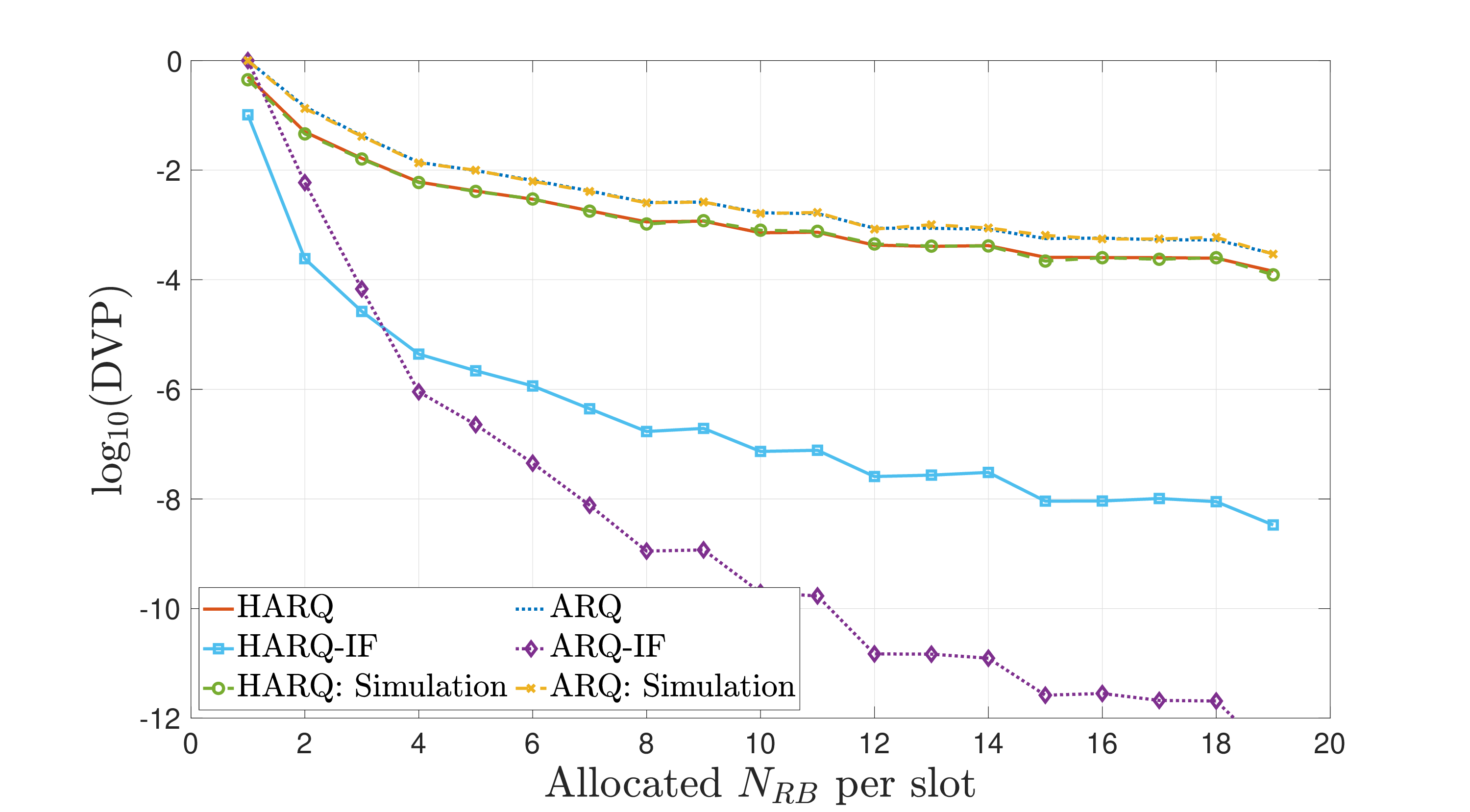}
\caption{Performance comparison of the proposed DVP evaluation schemes with the immediate feedback (IF) schemes for default configuration. The simulated DVP is also shown.}
\label{fig.Comparison}
\end{figure}
\begin{figure*}[t]
\centering
\begin{subfigure}[t]{0.49\linewidth}
\centering
\includegraphics[width=\textwidth]{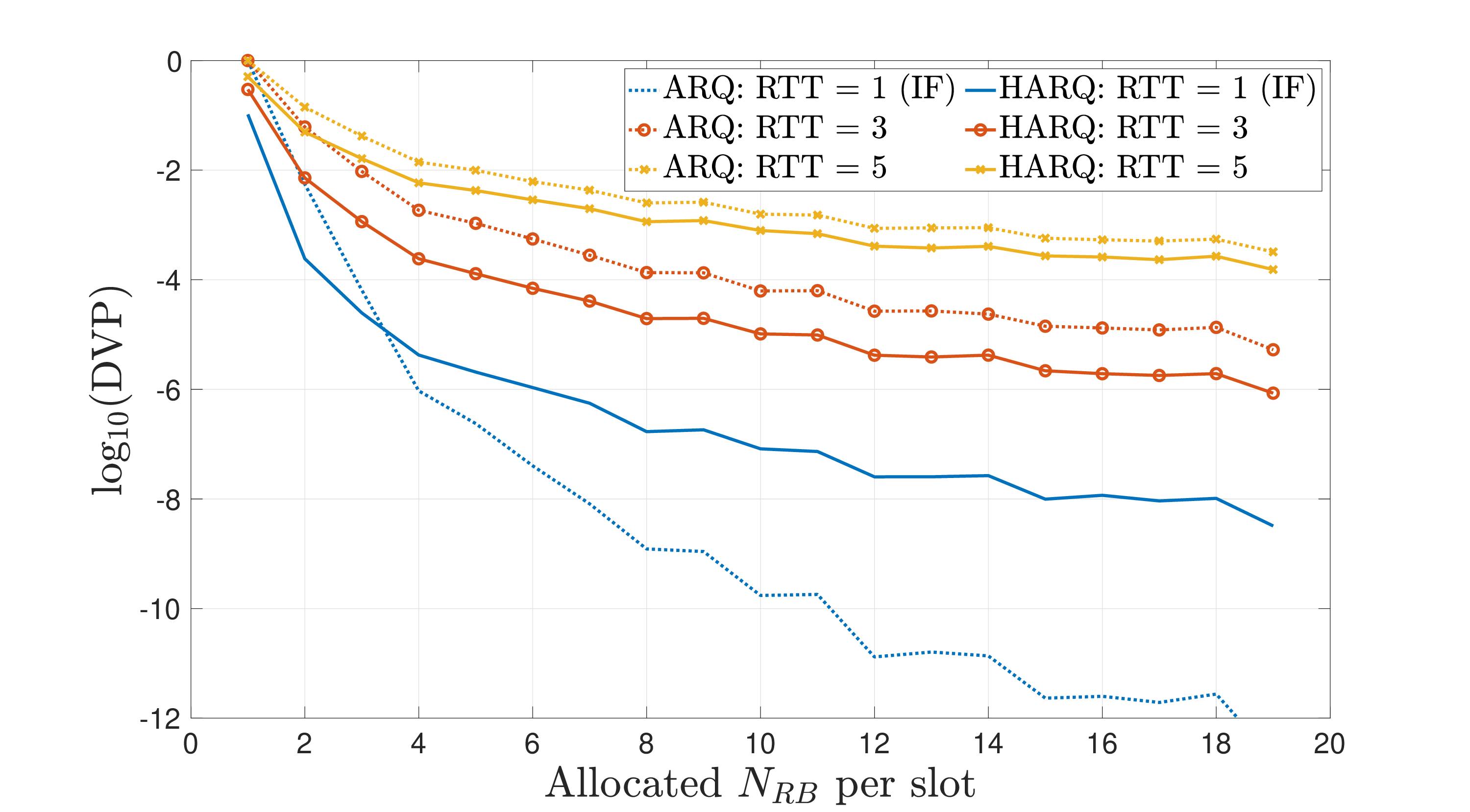}
\caption{DVP vs. $\nrb$ for different RTT. $d=8.5$ ms.}
\label{fig:Sim_dvp-nrb-rtt}
\end{subfigure}
\begin{subfigure}[t]{0.49\linewidth}
\centering
\includegraphics[width=\textwidth]{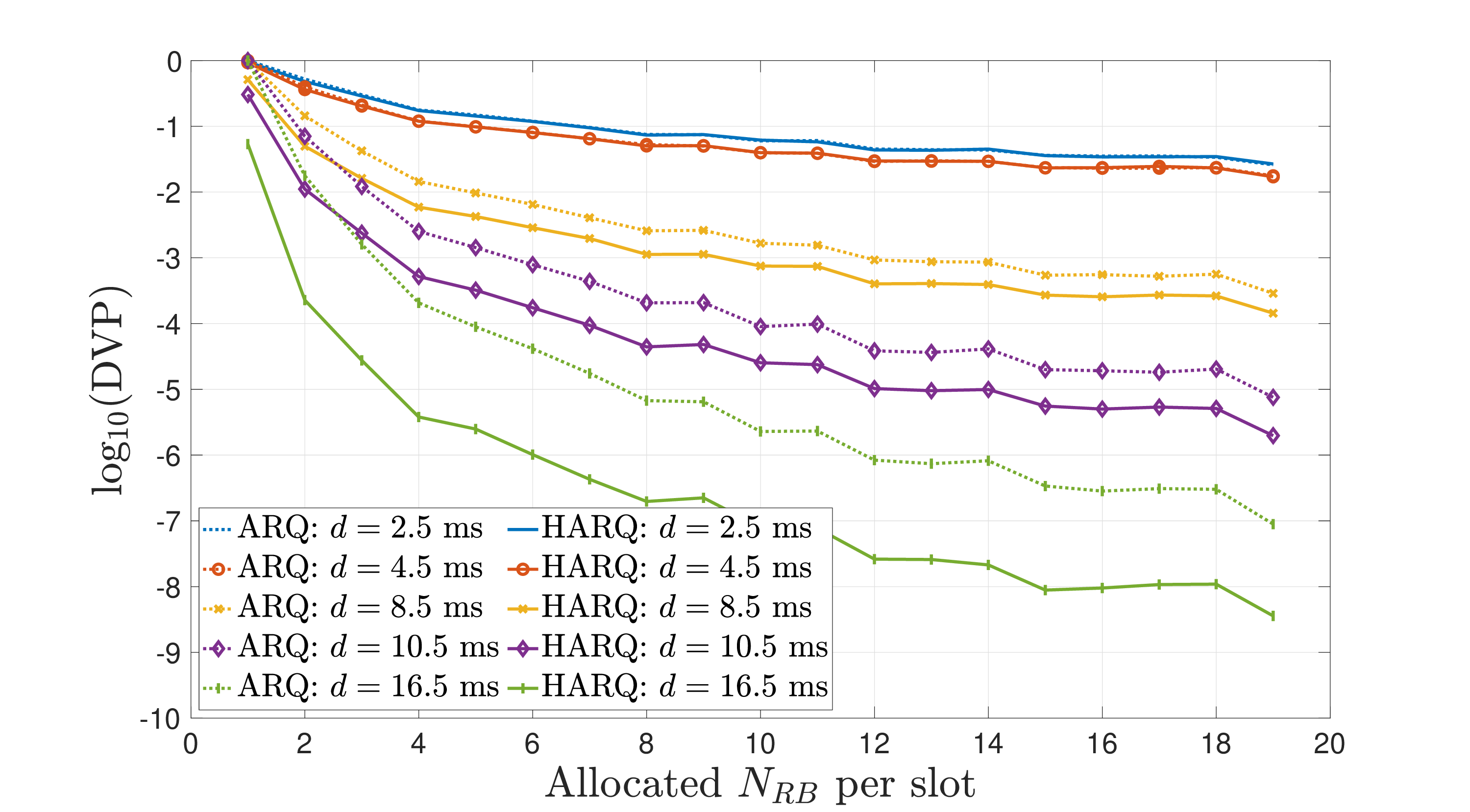}
\caption{DVP vs. $\nrb$ for different target delay $d$. RTT $=4$ slots.}
\label{fig:Sim_dvp-nrb-d}
\end{subfigure}
\caption{DVP vs. allocated $\nrb$ per slot for different delay parameters, namely RTT and $d$.}
\label{fig:Sim_dvp-nrb-nAndRtt}
\end{figure*}

Now, we move on to the MCS selection mechanism. 
The MCS selection follows 3GPP standards~\cite{3gpp.38.214}, where we choose an MCS index from 0 to 28, and obtain the modulation order, coding rate and spectral efficiency $\eta$ corresponding to it.
This range of MCS corresponds to a lower and upper limit of possible $\nrb$ for a given uncoded packet length $n$.
To minimize DVP for a given $\nrb$, the smallest MCS with $180\nrb\eta\geq n$ is chosen, that is, an MCS capable of supporting $n$ bits on $\nrb$ resources. 

We compute PER for ARQ and HARQ using Monte Carlo methods as described in Section~\ref{sec:model}. 
For DVP calculation, we use \eqref{dvpbound} for ARQ and ALGORITHM~\ref{Alg:getWaitProbability} with \eqref{eq:secIR_waitDelay} and \eqref{eq:secIR_totDelay} for HARQ. Note that we observe the DVP changing in steps at various points in all figures. This results from the finite and discrete nature of MCS selection and RB allocation in the 5G standard.

\subsubsection*{Performance Comparison}
To evaluate performance, we compare the proposed methods with state-of-the-art single-server models, which are accurate only under the assumption of zero decoding and feedback delays, thereby eliminating the need for multiple processes. In this section, these models are referred to as the immediate feedback models (IF), where feedback is assumed to be available immediately after the transmission slot, effectively setting RTT to one slot.
The IF serves as the benchmark because the key novelty of this work lies in addressing the unrealistic assumptions of zero RTT and single-process ARQ/HARQ implementations.
We use two IF benchmarks, ARQ-IF and HARQ-IF, derived by setting $\ddecod = \delta = 0$ in the ARQ and HARQ schemes, respectively.

In Fig.~\ref{fig.Comparison}, we present the DVP for ARQ, HARQ, and IF under the default parameter settings, with the HARQ results validated using an event-based numerical simulation in MATLAB. The simulation confirms that, under the model assumptions, HARQ accurately computes the DVP across different values of $\nrb$, reflecting a realistic 5G scenario.
Next, we observe that ARQ consistently produces a worse DVP than HARQ, as expected, due to the absence of incremental redundancy, a key feature of modern 5G systems. This trend holds throughout the section, except for extremes such as those observed here for IF, where it only holds for $\nrb<3$. For  $\nrb\geq4$, a different trend results from the 1-slot RTT that allows up to 8 attempts within the default target delay of 8.5 ms. While ARQ can fully utilize these opportunities, HARQ is constrained to $M=4$ retransmissions.
Finally, IF shows a DVP that is 4 to 6 orders of magnitude smaller than that of HARQ, which is the price one has to pay for the added delay. This shows that IF comes with a large sacrifice in terms of DVP accuracy if used to approximate the DVP of a realistic 5G HARQ.
This comparison also highlights that ARQ is a much better closed-form approximation for HARQ than the IF.

In the remainder of this section, we evaluate the DVP trends across various parameters. As the accuracy improvement with respect to IF remains consistent across configurations, we focus only on the proposed ARQ and HARQ schemes to maintain clarity.
\begin{figure}[t]
\centering
\includegraphics[width=0.98\linewidth]{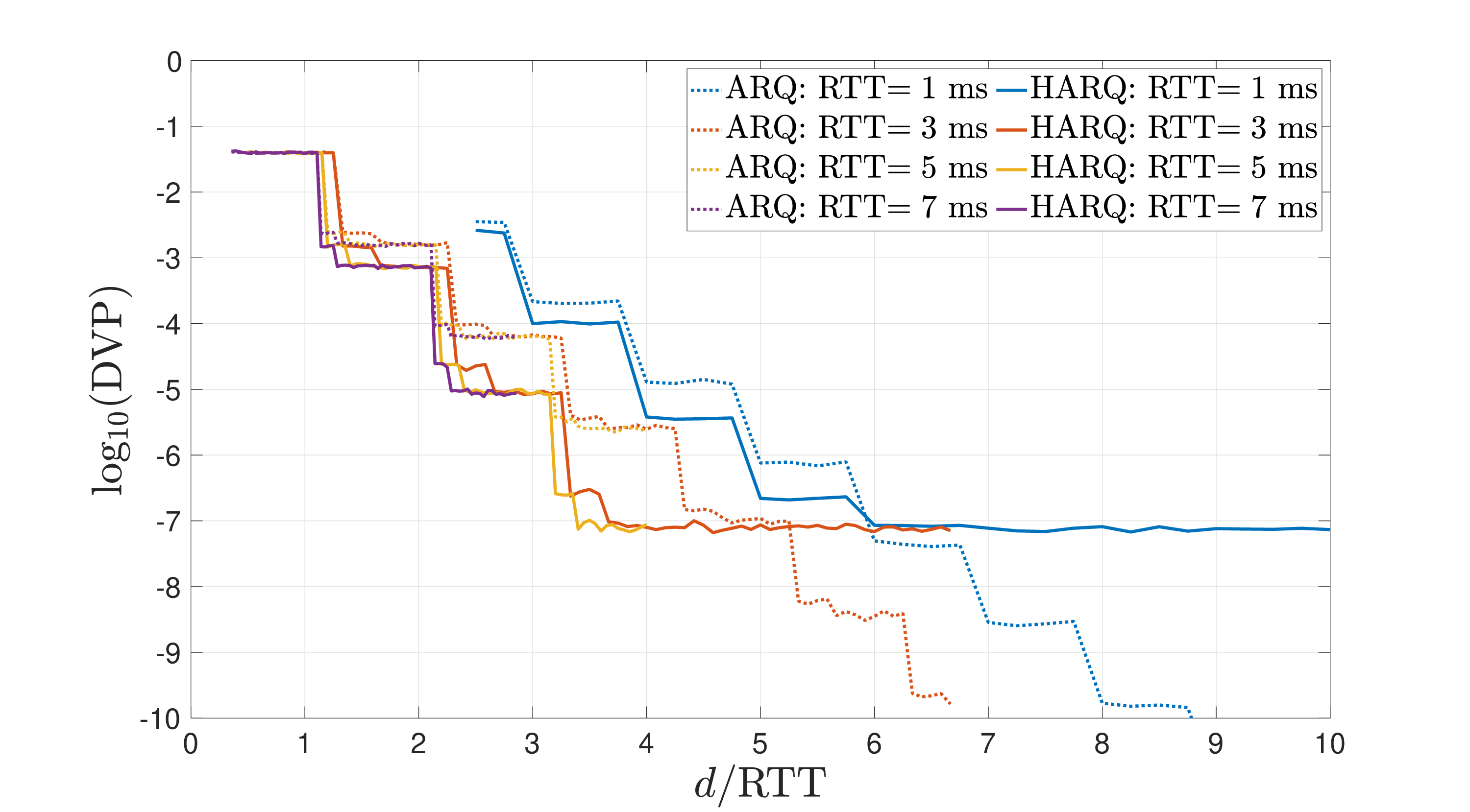}
\caption{DVP vs. $\nicefrac{d}{1+\ddecod+\delta}$, the delay target to RTT ratio.}
\label{fig:Sim_dvp-dByRtt}
\end{figure}

\begin{figure*}[t]
\centering
\begin{subfigure}[t]{0.49\textwidth}
\centering
\includegraphics[width=\textwidth]{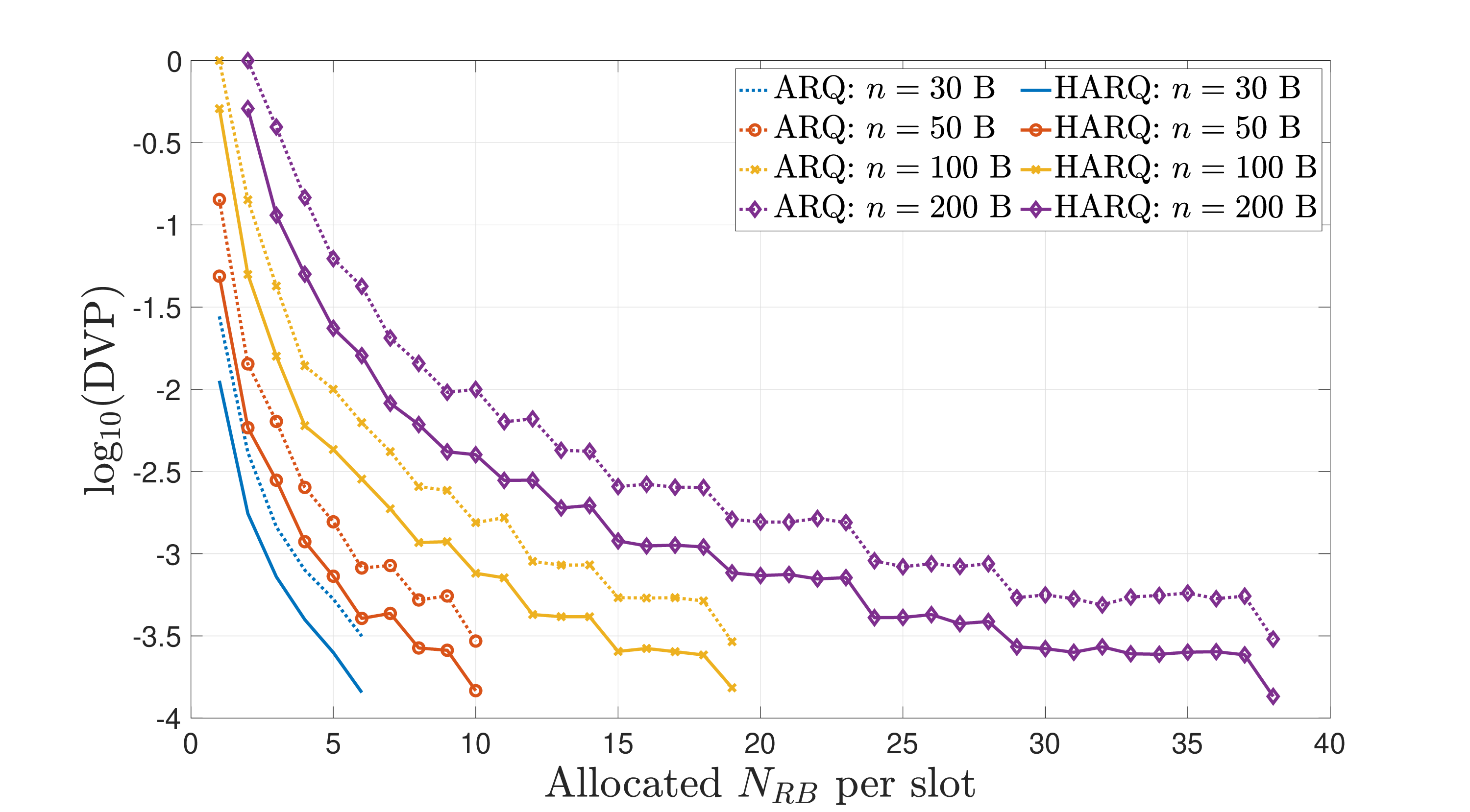}
\caption{DVP vs. allocated $\nrb$ per slot.}
\label{fig:Sim_dvp-nrb-n}
\end{subfigure}
\begin{subfigure}[t]{0.49\textwidth}
\centering
\includegraphics[width=\textwidth]{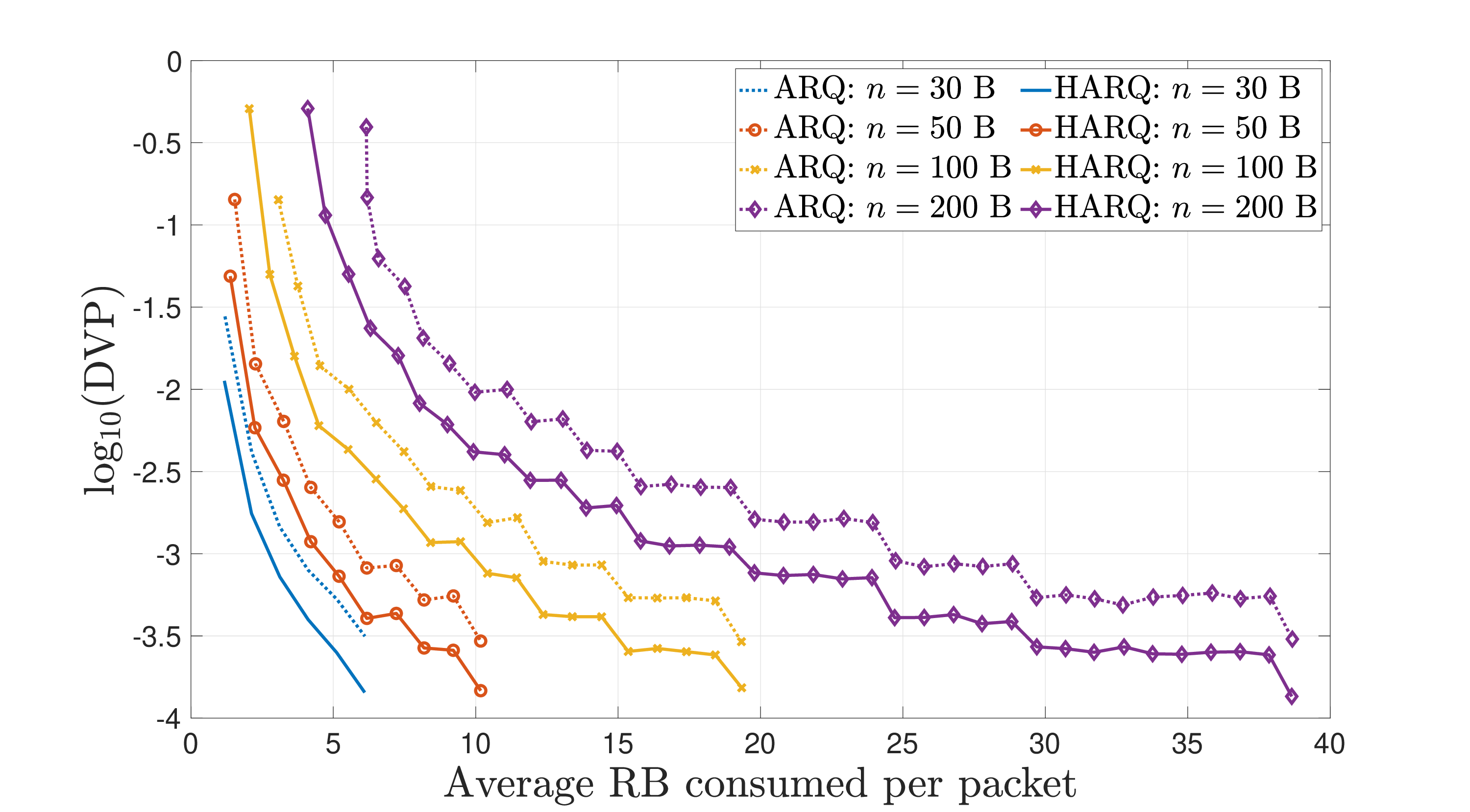}
\caption{DVP vs. average consumed $\nrb$ per packet.}
\label{fig:Sim_dvp-nrbUsed-n}
\end{subfigure}
\caption{DVP vs. resource blocks $\nrb$ for different uncoded packet lengths $n$.}
\label{fig:Sim_dvp-nrb-nAndUsedn}
\end{figure*}

\subsubsection*{The Effect of Delay Parameters:}
Now, we examine the impact of various delay components in the DVP of ARQ and HARQ.
In Fig.~\ref{fig:Sim_dvp-nrb-rtt} and Fig.~\ref{fig:Sim_dvp-nrb-d}, we show the DVP as a function of allocated $\nrb$ per slot across different RTTs and target delays.
Note that the RTT=1 corresponds to the IF approximation. 
We focus on two key observations. 
First, we observe that RTT substantially influences DVP, and a larger RTT increases the significance of this work over IF assumptions, especially with larger resource allocations.
This is because a large RTT could quickly eat into the delay margin with each retransmission. We also observe that the performance gap between ARQ and HARQ increases when the delay margin is not tight, corresponding to a larger target delay or a smaller RTT.
Second, the improved DVP through a larger $\nrb$ (and a lower MCS resulting from it) becomes more pronounced with a longer delay target. 
For example, the 3-order magnitude DVP improvement between $\nrb=1$ (MCS-28) and $\nrb=19$ (MCS-0) at a default $d=8.5$ ms grows to 7 orders at $d=16.5$ ms, i.e., doubling the delay target provides an additional DVP improvement of up to 4 orders.

Having observed that an increase in target delay or a decrease in RTT similarly affects DVP, it becomes useful to assess DVP as a function of their ratio.
To this end, we show DVP against the target delay-to-RTT ratio $\frac{d\times10^3}{1+\ddecod+\delta}$ in Fig.~\ref{fig:Sim_dvp-dByRtt}, for a fixed $\nrb$ of 10 and a packet length of 100 Bytes, corresponding to MCS-3. 
The RTT is converted to milliseconds for consistency, with 1 ms slots in the chosen numerology of 0 (see Section.\ref{sec:model}).
To obtain this ratio, we fix RTT at various values, as depicted, and vary the target delay from 2.5 to 10 ms.
The plot shows a consistent improvement of approximately one order of magnitude in DVP per unit increase in the ratio across RTT values. 
Another interesting observation comes from the comparison of ARQ and HARQ. While HARQ understandably provides better DVP performance in general, it saturates at around $10^{-7}$ for this default configuration. 
This limitation arises because HARQ, unlike ARQ, is restricted in its retransmission attempts and thus cannot fully leverage larger delay margins, as seen by comparing $k_d$ from \eqref{eq:secARQ_kd} and \eqref{eq:secIR_kd}.
\begin{figure}[t]
\centering
\includegraphics[width=0.98\linewidth]{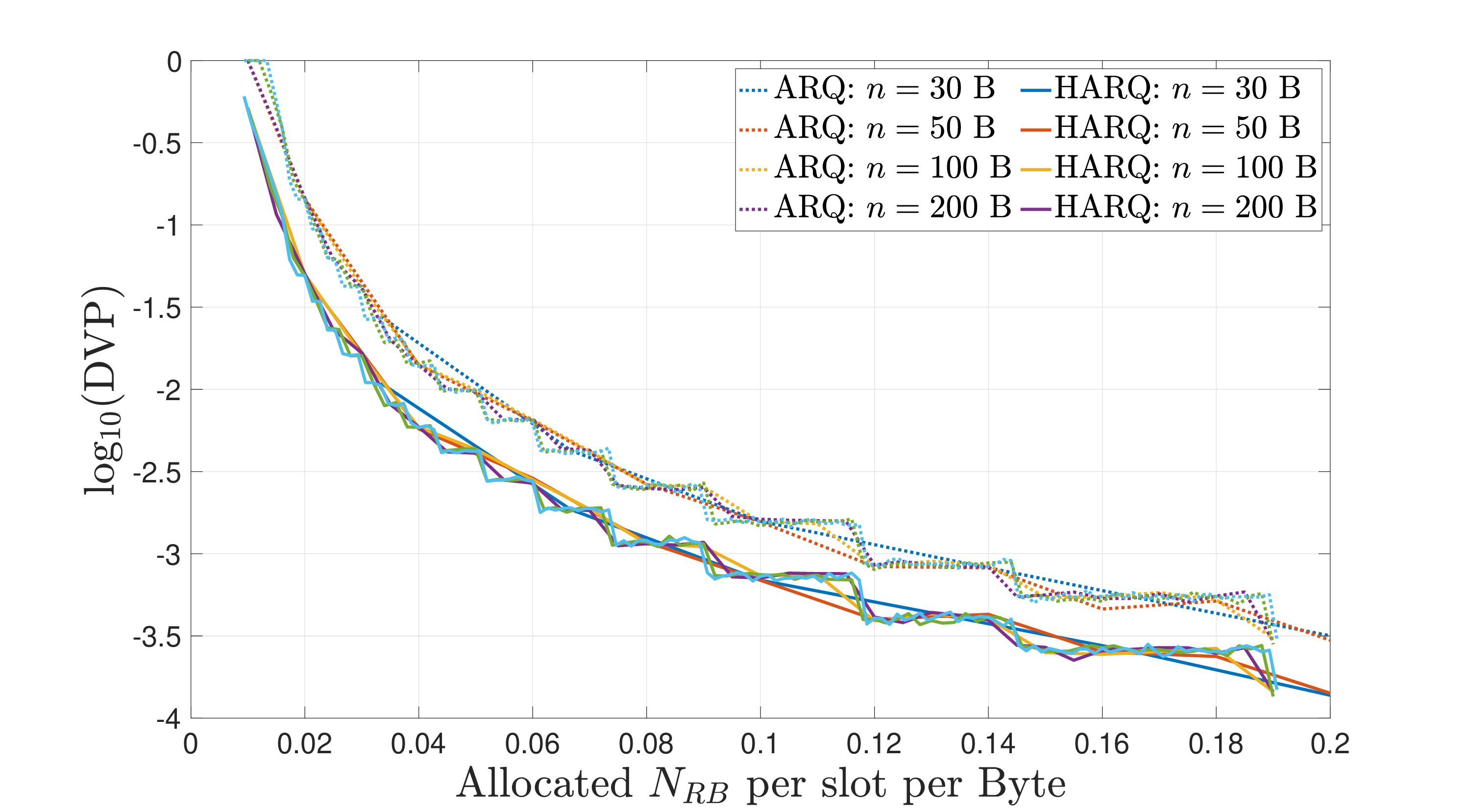}
\caption{DVP vs. allocated resources per slot per Byte of packet length, $\nicefrac{8\nrb}{n}$ for different $n$.}
\label{fig:Sim_dvp-nrb-nrbByn}
\end{figure}

\subsubsection*{The Effect of Resource Allocation:}
Now, we study the effect of packet length and resource allocation in DVP. In Fig.~\ref{fig:Sim_dvp-nrb-n}, we show the DVP variation with allocated $\nrb$ per slot for different uncoded packet lengths. 
As seen already, increasing $\nrb$ generally improves DVP, and the improvement rate is significant, providing up to a four-order reduction in DVP across the $\nrb$ range, corresponding to an equivalent MCS range from 28 down to 0. This range of $\nrb$ depends on the packet length. Thus, as observed in the figure, this DVP reduction with respect to $\nrb$ is much steeper for smaller packet sizes.

Note that when the PER is high, resulting from a frugal resource allocation, the number of retransmissions required for success also gets higher.
This increases the average number of resources consumed per packet as illustrated in Fig.~\ref{fig:Sim_dvp-nrbUsed-n}, where DVP is plotted against the average resource consumption per packet for different packet lengths.
We used the same data from Fig.~\ref{fig:Sim_dvp-nrb-n} for comparison, and one can observe that the curve gets steeper for smaller $\nrb$. 
This effect becomes extreme for persistent ARQ with unlimited retransmission attempts, where the expected number of attempts is given by ${1}/{1-p}$.

In Fig.~\ref{fig:Sim_dvp-nrb-nrbByn}, we analyze the combined effect of $\nrb$ and $n$ by studying resource allocation normalized by packet length. The DVP is plotted against $\nrb$ per byte for different fixed packet lengths. 
Notably, the plots for different packet lengths align closely, indicating that the DVP depends primarily on the resources allocated per byte rather than on the individual values of $\nrb$ or $n$. This insight shows that with properly allocating resources, lower DVP can be achieved even for larger packets.
\begin{figure*}[t]
\centering
\begin{subfigure}[t]{0.49\textwidth}
\centering
\includegraphics[width=\textwidth]{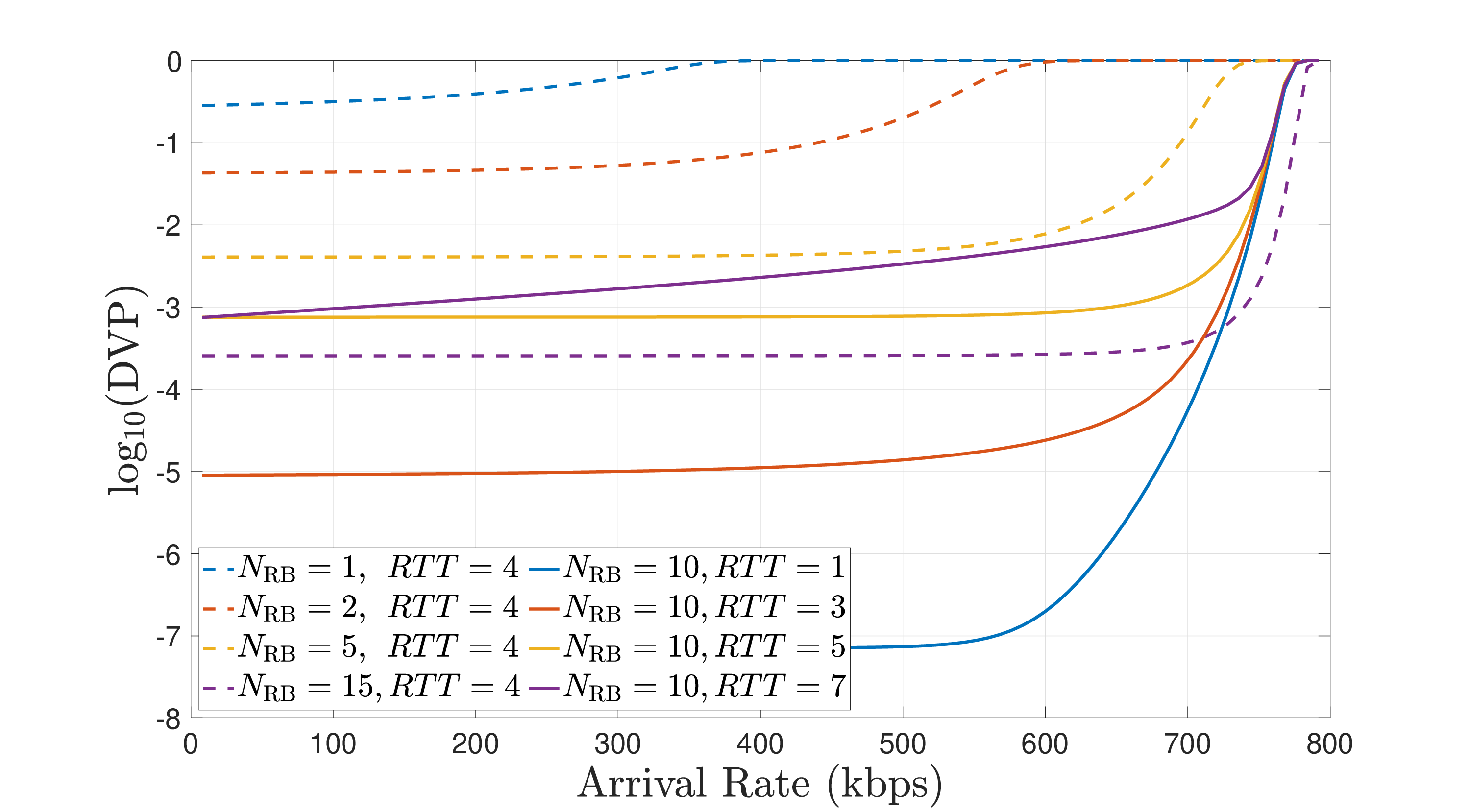}
\caption{DVP vs. of arrival rate.}
\label{fig:DvpvsArrivalrate}
\end{subfigure}
\begin{subfigure}[t]{0.49\textwidth}
\centering
\includegraphics[width=\textwidth]{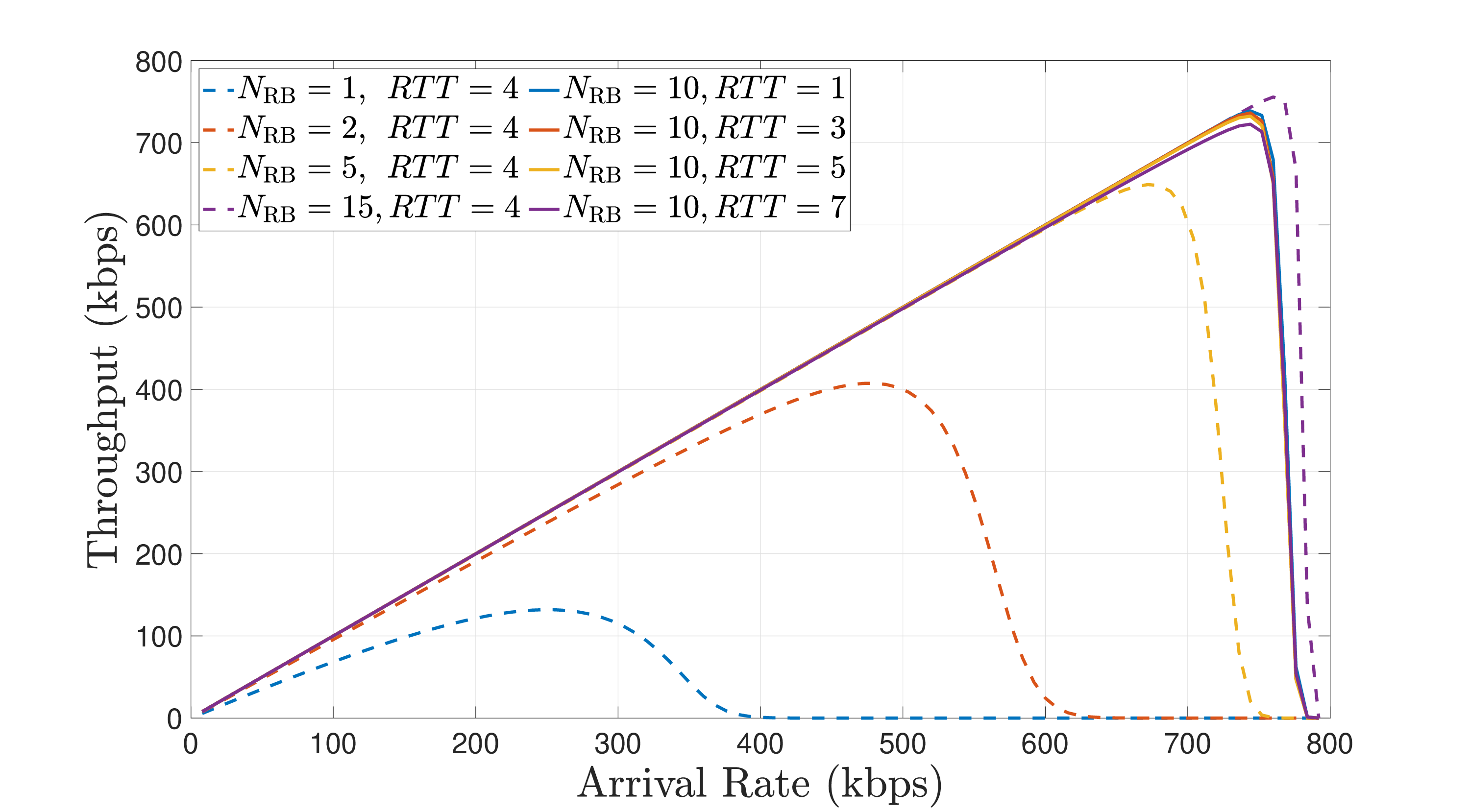}
\caption{Throughput vs. arrival rate showing the optimums.}
\label{fig:ThroughputVsArrivalrate}
\end{subfigure}
\caption{DVP and throughput vs. arrival rate for HARQ with different RTT and $\nrb$ with a fixed $n$ and varying $f$.}
\label{fig:DvpThroughputVsArrivalRate}
\end{figure*}


\subsubsection*{Effect of Arrival Rate:}
A 5G system with strict latency requirements must discard all packets that violate the target delay, leaving only the remaining packets to contribute to the throughput. 
The throughput thus depends primarily on the arrival rate and the DVP. 
To study this relationship, we show the DVP and throughput of HARQ as a function of the arrival rate in Fig.~\ref{fig:DvpThroughputVsArrivalRate}.
Here, we vary the arrival rate $fn/T$ by adjusting $f$, while keeping the packet length $n$ fixed at its default value.
The dotted lines correspond to variations in $\nrb$ for a fixed RTT, while the solid lines represent variations in RTT for a fixed $\nrb$.

In Fig~\ref{fig:DvpvsArrivalrate}, observe the region of arrival rate at which the DVP rises sharply toward 1 from its asymptotic lower bound.
This rise in DVP lowers the throughput, leading to the emergence of an optimal arrival rate that maximizes the throughput, as illustrated in Fig.~\ref{fig:ThroughputVsArrivalrate}.
Notably, while RTT is one of the key parameters deciding the DVP, the arrival rate at which the DVP goes to a very high value (ca. $750$ kbps in this example), and thus the optimum arrival rate, appears largely independent of the RTT. This, however, is not the case for different $\nrb$, where the optimum arrival rate increases with more resource allocation. 
These observations are useful in the resource allocation tailored for RTT and packet length.

\section{Conclusion}\label{sec:conclusion}
In this work, we aimed to characterise the QoS in a 5G system focusing on ARQ and HARQ-IR retransmission schemes by accurately evaluating the delay violation probability (DVP) for a given target delay. 
Unlike existing methods, we proposed a novel delay model that incorporated decoding and feedback delay into it. This also demanded the inclusion of a multi-server queueing model with multiple parallel ARQ/HARQ processes where the packets do not wait for feedback from previous transmissions, thereby saving valuable transmission opportunities. 
Using this delay model and a novel packet error rate (PER) model based on finite blocklength packet transmission theory, we computed closed-form expressions and algorithms to compute DVP for ARQ and HARQ schemes. 
Our assumptions closely followed 3GPP standards and can be adapted to various scenarios, thus enhancing the usability of this work.

Our numerical evaluations demonstrated that the proposed evaluation schemes significantly outperform state-of-the-art immediate feedback (IF) models in terms of accuracy, with the performance gap widening as decoding and feedback delays increase. 
We observed that While HARQ achieves better DVP outcomes than persistent ARQ under normal circumstances, persistent ARQ is better in some specific cases due to the practical constraints of allowing an arbitrary number of retransmission attempts for HARQ. 
We illustrated how parameter tuning affects DVP and emphasised the importance of balancing MCS and resource allocation to regulate QoS in 5G networks. 
We observed that sufficient resource allocation per byte of packet size can help achieve low DVP levels, even for larger packet sizes. 
Additionally, we saw that the throughput of the system initially increases with the arrival rate but eventually decreases due to the increase in delay violation. This revealed the existence of an optimum arrival rate that maximises the throughput.

Beyond DVP analysis, our numerical results can inform resource allocation algorithms, enabling them to guarantee QoS under specific system configurations. 
These findings underscore the importance of optimizing resource allocation and MCS selection to meet the stringent delay and reliability requirements of latency and reliability sensitive 5G applications, marking a step toward real-world implementation of 5G networks. 

\bibliographystyle{IEEEtran}
\bibliography{refs.bib}


\end{document}